\documentclass[11pt]{article}
\usepackage[a4paper,margin=1in]{geometry}
\usepackage{amsmath,amssymb,mathtools,amsthm}
\usepackage{braket}
\usepackage[T1]{fontenc}
\usepackage{lmodern,microtype,graphicx}
\usepackage[font=small,labelfont=bf]{caption}

\usepackage[hidelinks,hypertexnames=false]{hyperref}
\usepackage{booktabs}
\usepackage{array,placeins}
\newcommand{\dd}{\,\mathrm d}

\newcommand{\eps}{\epsilon}
\DeclareMathOperator{\Tr}{Tr}
\newtheorem{theorem}{Theorem}
\newtheorem{proposition}[theorem]{Proposition}
\newtheorem{corollary}[theorem]{Corollary}
\newtheorem{lemma}[theorem]{Lemma}

\newcommand{\KL}{D_{\mathrm{KL}}}
\hypersetup{pdftitle={Wigner entropy below vacuum: physical counterexamples and stability limits},pdfauthor={Zixuan He}}
\begin{document}
\begin{center}
{\Large\bfseries Wigner entropy below vacuum: physical counterexamples and stability limits}\par
\vspace{0.55em}
{\large Zixuan He}\par
University of Glasgow, Glasgow, United Kingdom\par
24 September 2026
\end{center}

\begin{abstract}
Positive Wigner functions can have less Shannon entropy than the vacuum, even arbitrarily close to the vacuum state. We construct physical counterexamples and identify the competition that controls their entropy: the relative entropy to the vacuum phase-space density can exceed twice the mean photon number. A two-level family exhibits a finite window in which coherence lowers entropy while preserving global Wigner positivity. A complementary construction repairs a remote negative tail with an exponentially small thermal admixture; an explicit mixing weight of $2\times10^{-21}$ preserves a rigorously established entropy decrease. Optimisation over all one-mode Wigner-nonnegative states gives the sharp low-energy scale $E^\gamma/[\ln(1/E)]^\beta$, with $\gamma\simeq0.7412$ and $\beta\simeq0.5861$. Because $\gamma<1$, tensor products can have vanishing total energy and trace distance from vacuum while their entropy deficit diverges. We derive the exact half-transmission threshold for universal Shannon-entropy recovery and connect it to companion results showing residual non-Gaussian structure. The counterexamples reveal distinct controls on entropy: coherence sets the local descent, mode number amplifies it, and loss restores the vacuum bound.
\end{abstract}

\section{Why positive phase-space densities can lie below vacuum}
\label{sec:central}

The vacuum is the reference state for optical noise, yet proximity to vacuum does not by itself control phase-space entropy. A physical state may have an everywhere positive Wigner function, arbitrarily little excitation and Shannon entropy below the vacuum value. Across a growing number of modes, the discrepancy becomes unbounded even as the complete state approaches vacuum in trace distance. The counterexamples developed here expose the mechanism behind this failure and determine its sharp single-mode energy scale.

We use $[\hat q,\hat p]=i$, natural logarithms and
\begin{equation}
W_\rho(q,p)=\frac1\pi\int_{\mathbb R}e^{2ipy}
\langle q-y|\rho|q+y\rangle\,\dd y,\qquad
W_0(q,p)=\pi^{-1}e^{-q^2-p^2}.
\label{eq:wigner-def}
\end{equation}
For $W_\rho\ge0$, let $h(W_\rho)=-\int W_\rho\ln W_\rho$.
The vacuum has $h_0=1+\ln\pi$. The conjecture of Van Herstraeten and Cerf proposed $h(W_\rho)\ge h_0$ for every physical Wigner-positive state \cite{VHC2021}. Positive results for Fock qubits, beam-splitter states and mixed-state subclasses made the physical constraint on these probability densities central to the question \cite{DiasPrata2023,QianGagatsos2024,VHC2025,QianGagatsos2026}.

An exact identity isolates the competing contributions. With complete mean photon number $E=\Tr(\hat N\rho)$,
\begin{equation}
\boxed{\ h(W_\rho)-h_0=2E-\KL(W_\rho\Vert W_0)\ },
\qquad
\KL(W_\rho\Vert W_0)=\int W_\rho\ln\frac{W_\rho}{W_0}.
\label{eq:central-balance}
\end{equation}
Indeed, $-\ln W_0=\ln\pi+q^2+p^2$ and $\int(q^2+p^2)W_\rho=1+2E$. Thus a counterexample requires a physical, nonnegative Wigner density whose relative entropy exceeds its excitation contribution. Finite energy makes both sides well defined (Appendix~\ref{app:sharp-statements}).

The balance admits two physical controls: reducing off-diagonal coherence at fixed populations, or adding a broad positive tail at exponentially small weight. Both preserve the density-operator constraints while allowing the relative-entropy term to exceed $2E$. Their low-energy limits reveal how entropy depends on coherence, complete excitation cost and mode number.

\section{A counterexample window at fixed energy}
\label{sec:coherence-main}

Consider the physical two-level family
\begin{equation}
\sigma_{n,t,\lambda}=
\frac{|0\rangle\langle0|+t^2|n\rangle\langle n|
+\lambda t(|0\rangle\langle n|+|n\rangle\langle0|)}
{1+t^2},\qquad 0\le\lambda<1.
\label{eq:main-doublet}
\end{equation}
Its energy $E=nt^2/(1+t^2)$ is independent of the coherence fraction $\lambda$.
For every fixed $n$ and $\lambda<1$, sufficiently small $t$ gives strict Wigner positivity throughout phase space, and
\begin{equation}
h(W_{\sigma_{n,t,\lambda}})-h_0
=(2n-2^n\lambda^2)t^2+O_{n,\lambda}(t^4).
\label{eq:main-coefficient}
\end{equation}
Theorem~\ref{thm:general-doublet} proves both the global positivity and the expansion. The coefficient compares an energy cost linear in Fock level with a coherence contribution growing as $2^n$. Its first negative quadratic direction occurs at $n=3$, for $\lambda>\sqrt3/2$.

At $n=3$, a fully explicit interval is available:
\begin{equation}
0<|t|\le\frac1{200},\quad\lambda=\frac9{10}
\quad\Longrightarrow\quad
W_\sigma\ge\frac{1/10-1801t^2}{1+t^2}W_0>0,\quad
h(W_\sigma)-h_0<-\frac{47}{100}t^2.
\label{eq:main-certified}
\end{equation}
These inequalities follow from polynomial bounds and require no entropy quadrature (Theorem~\ref{thm:explicit-doublet}). At fixed $t$, increasing $\lambda$ strengthens the relative-entropy term without changing the energy. Figure~\ref{fig:window} shows that the vacuum entropy is crossed before the Wigner function loses positivity.

\begin{figure}[!htbp]
\centering\includegraphics[width=\linewidth]{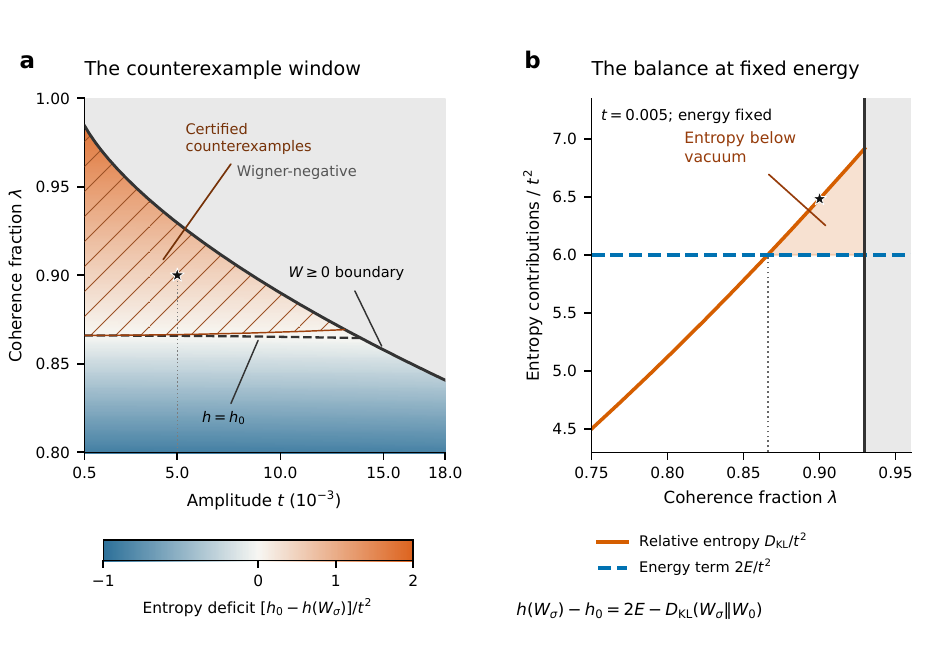}
\caption{\textbf{Coherence opens a window of positive-Wigner counterexamples.}
\textbf{a}, The vacuum--three-photon family of Eq.~\eqref{eq:main-doublet}.
The solid curve is the exact global positivity boundary $\lambda_{\rm pos}(t)$; grey denotes Wigner negativity somewhere in phase space. The dashed curve is the numerical entropy-zero contour. Colour gives $[h_0-h(W_\sigma)]/t^2$, with separate linear scales on the two sides of zero. Hatching marks the sufficient counterexample region proved by a polynomial entropy inequality in Appendix~\ref{app:window}. The star, $t=0.005$, $\lambda=0.9$, lies in the explicit interval~\eqref{eq:main-certified}.
\textbf{b}, At $t=0.005$, changing $\lambda$ leaves both populations and the energy fixed. The relative-entropy term crosses $2E$ at $\lambda\simeq0.86584$, while Wigner positivity persists up to $\lambda\simeq0.92959$. The shaded gap is the entropy decrease. Colour and the dashed zero contour are numerical evaluations; the hatched sufficient region and solid positivity boundary have analytic definitions.}
\label{fig:window}
\end{figure}

The parameter map follows from minimising a one-dimensional function:
\begin{equation}
\lambda_{\rm pos}(t)=\min_{x>0}
\frac{\sqrt3[1+t^2F(x)]}{4t x^{3/2}},
\qquad F(x)=-1+6x-6x^2+\frac43x^3.
\label{eq:main-window}
\end{equation}
For $0<t<1$, the minimum is attained at the unique positive root of
$t^2(4x^3-6x^2-6x+3)=3$.
For the plotted amplitudes, the sufficient entropy condition is
\begin{equation}
\lambda>\lambda_{\rm cert}(t):=
\left[\frac{6-110t^2+(86348/3)t^4}{8-488t^2}\right]^{1/2},
\qquad \lambda<\lambda_{\rm pos}(t).
\label{eq:main-window-cert}
\end{equation}
Appendix~\ref{app:window} derives these thresholds and the strict decrease of entropy with positive $\lambda$.

\section{A remote negative tail can be repaired at negligible mixing cost}
\label{sec:repair-main}

Thermal mixing gives a second route to the physical constraint. Set
\begin{equation}
|\psi_t\rangle=\frac{|0\rangle+t|3\rangle}{\sqrt{1+t^2}},
\qquad
\tau=\frac23\sum_{j=0}^{\infty}3^{-j}|j\rangle\langle j|,
\qquad
\rho^{\rm th}_{t,\epsilon}=(1-\epsilon)|\psi_t\rangle\langle\psi_t|+\epsilon\tau.
\label{eq:main-thermal}
\end{equation}
The thermal density $W_\tau=(2\pi)^{-1}e^{-x/2}$ decays more slowly than the Gaussian factor of the coherent component, where $x=q^2+p^2$. Its energy cost is included in
$E=3(1-\epsilon)t^2/(1+t^2)+\epsilon/2$.

Angular minimisation determines exactly how much repair is needed:
\begin{align}
M(t)&=\max_{x\ge0}\frac{2e^{-x/2}}{1+t^2}
\left[-1-t^2F(x)+\frac{4t}{\sqrt3}x^{3/2}\right]_+,\notag\\
\epsilon_*(t)&=\frac{M(t)}{1+M(t)},\qquad
\lim_{t\downarrow0}t^{2/3}\ln\epsilon_*(t)
=-\frac12\left(\frac34\right)^{1/3}.
\label{eq:main-repair}
\end{align}
The Wigner function is strictly positive for $\epsilon>\epsilon_*(t)$.
The negative region moves to squared radii of order $t^{-2/3}$; the broad thermal tail covers it at exponentially small weight (Proposition~\ref{prop:repair}).

For the concrete choice
\begin{equation}
t=\frac{\sqrt3}{2000},\qquad \epsilon_{\rm f}=2\times10^{-21},
\qquad \rho_{\rm f}=\rho^{\rm th}_{t,\epsilon_{\rm f}},
\label{eq:main-new-example}
\end{equation}
we obtain the global certificate
\begin{equation}
W_{\rho_{\rm f}}>\frac{79}{10000}W_0>0,
\qquad h(W_{\rho_{\rm f}})-h_0<0.
\label{eq:main-new-certified}
\end{equation}
Exact rational bounds prove positivity. Entropy concavity transfers the strict entropy decrease from the certified example in Appendix~\ref{app:fixed} to the smaller repair in Eq.~\eqref{eq:main-new-example}; the full argument is in Appendix~\ref{app:tight-repair}. Numerical integration gives an entropy difference of approximately $-1.5001440\times10^{-6}$ at complete energy $2.2499983\times10^{-6}$. The repair weight is about $2.21$ times the minimum required for this fixed core and this thermal density.

\begin{figure}[!htbp]
\centering\includegraphics[width=\linewidth]{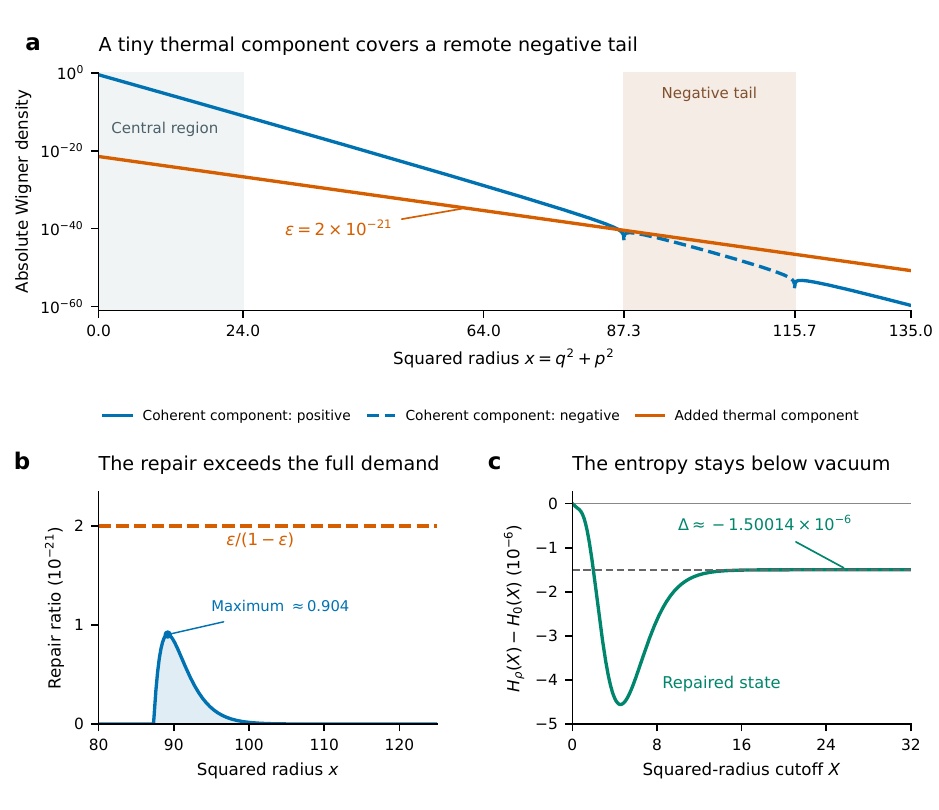}
\caption{\textbf{An extremely small thermal admixture repairs the tail while preserving entropy descent.}
The source is Eq.~\eqref{eq:main-new-example}.
\textbf{a}, Weighted Wigner components along a ray with $\cos3\theta=-1$. The logarithmic axis shows absolute magnitudes: solid blue is the positive coherent component and dashed blue its negative part. Orange is the added thermal density. The coherent component is negative only for $87.3116\ldots<x<115.6901\ldots$ and is dominated there by the thermal contribution.
\textbf{b}, The local demand $M_x=2e^{-x/2}[-1-t^2F(x)+(4t/\sqrt3)x^{3/2}]_+/(1+t^2)$ is everywhere below the available mixing odds $\epsilon_{\rm f}/(1-\epsilon_{\rm f})$. Its maximum is approximately $9.03978\times10^{-22}$.
\textbf{c}, The cumulative entropy difference of the actual repaired state, with
$H_\omega(X)=-\int_{q^2+p^2\le X}W_\omega\ln W_\omega\,\dd q\,\dd p$.
It tends to $\Delta=h(W_{\rho_{\rm f}})-h_0<0$. These are contributions to the full entropy integral, without renormalising the truncated density. Only the positive repaired state enters the entropy calculation. Appendix~\ref{app:tight-repair} supplies an all-plane positivity certificate and the analytic proof of the entropy sign.}
\label{fig:thermal}
\end{figure}

The same separation of central response and remote repair reveals a higher-order effect. For a vacuum--one-photon core,
$|\chi_t\rangle=(|0\rangle+t|1\rangle)/\sqrt{1+t^2}$, the weight
$\epsilon_t=4t^2\exp[-1/(4t^2)+1/(2t)]$ gives a strictly positive physical Wigner function and
\begin{equation}
h(W_{\omega_t})-h_0=-\frac43t^6+o(t^6),\qquad E_t=t^2+O(t^4).
\label{eq:main-one-photon}
\end{equation}
The thermal component occupies higher Fock levels. Theorem~\ref{thm:one-photon} proves this sixth-order descent, while Eq.~\eqref{eq:main-coefficient} characterises the separate limit with a fixed coherence fraction. Cubic coherence is the first quadratic instability of that fixed-fraction family; entropy descent also exists beyond the quadratic approximation.

\section{Sharp energy recovery fails to be uniform in dimension}
\label{sec:dimension}

The complete one-mode optimisation has a definite recovery law. Define
\begin{equation}
\mathcal D(E)=\sup_{\rho:\,W_\rho\ge0,\ \Tr(\hat N\rho)\le E}
[h_0-h(W_\rho)].
\label{eq:main-energy-sup}
\end{equation}
Theorem~\ref{thm:sharp-energy} proves
\begin{equation}
\boxed{\ \mathcal D(E)=\Theta\!\left(
\frac{E^\gamma}{[\ln(1/E)]^\beta}\right)\ },\qquad E\downarrow0,
\label{eq:main-sharp}
\end{equation}
where $\gamma\simeq0.7412033679$ and $\beta\simeq0.5861054961$ are specified exactly in Eqs.~\eqref{eq:critical-root}--\eqref{eq:sharp-constants}.
The upper bound covers every physical Wigner-nonnegative state with the stated energy, including infinite Fock support. A phase-averaged displaced squeezed repair gives the matching lower bound, with all displacement, squeezing and mixing energies counted.

The exponent reflects a competition between excitation and global positivity. Write $L=\ln(1/E)$. The upper-bound proof splits the vacuum--Fock coherences at a level $s$: energy controls the low-level contribution by $E2^s/s$, while positivity controls the remaining leading contribution by $s!/R^s$, with $R=L+\frac43\ln L$. Balancing their exponential rates at $s\sim r_*L$ gives
\begin{equation}
1-r_*\ln2=r_*(1-\ln r_*),\qquad
\gamma=1-r_*\ln2.
\label{eq:main-exponent-balance}
\end{equation}
The logarithmic exponent $\beta$ comes from the next terms: the factor $1/s$, the square-root factor in Stirling's formula and the logarithmic shift in $R$. Appendix~\ref{app:exponent-guide} works out this balance and shows how the physical repair attains it at
$n=\lfloor r_*L+c_*\ln L-K\rfloor$, with fixed $K$.
The deficit therefore tends to zero, but its ratio to $E$ diverges. There is no finite linear energy correction to the vacuum bound near zero excitation. The fixed one-photon construction has deficit of order $E^3$; the growing-level construction determines the optimal scale.

The exponent $\gamma<1$ has a direct many-mode consequence. Let
$\nu_m=|0\rangle\langle0|^{\otimes m}$,
$N_m=\sum_{j=1}^m\hat N_j$, and
\begin{equation}
\mathcal D_m(\rho)=mh_0-h(W_\rho),\quad
T(\rho,\nu_m)=\tfrac12\|\rho-\nu_m\|_1,\quad
d_{\mathcal F_m}(\rho)=\inf_{\sigma\in\mathcal F_m}\tfrac12\|\rho-\sigma\|_1,
\label{eq:multimode-definitions}
\end{equation}
where $\mathcal F_m$ is the trace-norm-closed convex hull of Gaussian states.
For every fixed $0<\zeta<(1-\gamma)/\gamma$, there are strictly Wigner-positive states with
\begin{align}
\Tr(N_m\rho_m)&\le m^{-\zeta}\longrightarrow0,\notag\\
d_{\mathcal F_m}(\rho_m)&\le T(\rho_m,\nu_m)\le m^{-\zeta/2}\longrightarrow0,\label{eq:dimension-distance}\\
\mathcal D_m(\rho_m)&\ge
C_1\frac{m^{1-\gamma(1+\zeta)}}{[(1+\zeta)\ln m]^\beta}
\longrightarrow+\infty.
\label{eq:dimension-deficit}
\end{align}

To see this, choose the one-mode attaining state $\sigma_{e_m}$ with
$e_m=m^{-1-\zeta}$ and set $\rho_m=\sigma_{e_m}^{\otimes m}$.
Energy and entropy add, giving the first and third bounds.
The operator inequality $N_m\ge I-\nu_m$ gives
$T(\rho_m,\nu_m)^2\le1-\langle0^{\otimes m}|\rho_m|0^{\otimes m}\rangle
\le\Tr(N_m\rho_m)$.
Finally $\nu_m\in\mathcal F_m$. This proves all three limits for the same states, directly from the one-mode theorem, as developed in the companion stability theory \cite{HeStability}.

\begin{figure}[!htbp]
\centering\includegraphics[width=\linewidth]{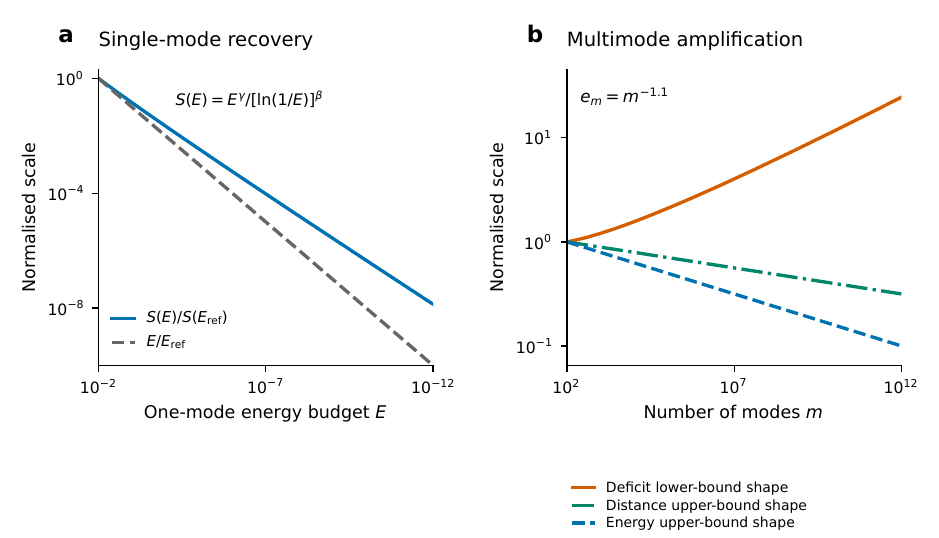}
\caption{\textbf{Single-mode entropy recovery becomes a dimensional instability.}
The curves show normalised analytic shape functions from Eq.~\eqref{eq:main-sharp} and Eqs.~\eqref{eq:dimension-distance}--\eqref{eq:dimension-deficit}.
\textbf{a}, The one-mode scale $E^\gamma/[\ln(1/E)]^\beta$ tends to zero more slowly than $E$. Each curve is divided by its value at $E_{\rm ref}=10^{-2}$.
\textbf{b}, With energy per mode $e_m=m^{-1.1}$, the upper bounds on total energy and distance from vacuum scale as $m^{-0.1}$ and $m^{-0.05}$; the lower bound on entropy deficit scales as $m^{0.184676\ldots}/(1.1\ln m)^\beta$. Each curve is divided by its value at $m_{\rm ref}=10^2$.
The curves show asymptotic dependence, with unspecified prefactors and finite onset scales.}
\label{fig:dimension}
\end{figure}

Every one-copy measurement event then has a probability difference from vacuum at most $T(\rho_m,\nu_m)$, while the entropy deficit grows without bound. The failure is a lack of dimension-independent continuity: closeness of the state and the accumulated phase-space entropy obey different scales.

\section{Loss restores the entropy bound before structure disappears}
\label{sec:loss-main}

Optical loss supplies a sharp second boundary. Let $\mathcal L_\eta$ denote pure loss of intensity transmissivity $\eta$, and let $\mathcal P_{m,E}$ contain Wigner-nonnegative states of complete input energy at most $E$. For every fixed $E>0$,
\begin{equation}
\boxed{\ 
\sup_{m\ge1}\ \sup_{\rho\in\mathcal P_{m,E}}
\mathcal D_m(\mathcal L_\eta^{\otimes m}\rho)=
\begin{cases}
0,&0\le\eta\le1/2,\\
+\infty,&1/2<\eta\le1 .
\end{cases}
\ }
\label{eq:main-loss}
\end{equation}
This threshold was established in Ref.~\cite{HeStability}; Appendix~\ref{app:loss} gives its complete short derivation from the present counterexamples and the Wehrl inequality.
At half transmission the output Wigner density is a rescaled Husimi density, whose entropy is bounded below by the vacuum value \cite{Wehrl}. Above half transmission, the one-mode coherence term becomes $\lambda^2(2\eta)^n$ while the energy cost remains linear in $n$. Choosing the level first and the amplitude second yields arbitrarily large multimode deficits within any fixed positive total-energy budget.

Entropy recovery does not force proximity to a Gaussian mixture. Wigner-positive states beyond the Gaussian convex hull, including their detection after loss, have an established literature \cite{FilipMista2011,Genoni2013}. A companion source family makes the separation quantitative at the recovery threshold \cite{HeStability,HeConcentration}. It consists of weak thermally repaired vacuum--three-photon modes with occupation $c/m$ and total input energy tending to $3c$. For $0<c<1$, its attenuated state $\widetilde\rho_{m,\eta}$ satisfies
\begin{equation}
\liminf_{m\to\infty}d_{\mathcal F_m}(\widetilde\rho_{m,\eta})
\ge c\eta^3e^{-c[1-(1-\eta)^3]},\qquad 0<\eta\le1 .
\label{eq:main-hull-gap}
\end{equation}
At $\eta=1/2$, the same state obeys the recovered entropy bound and retains a gap at least $(c/8)e^{-7c/8}>0$ (Fig.~\ref{fig:loss}). Appendix~\ref{app:residual} defines the source and reconstructs the positivity and witness proofs. To meet a prescribed input budget $E$, choose $c<\min(1,E/3)$. This source has bounded nonzero limiting energy; the sequence in Eqs.~\eqref{eq:dimension-distance}--\eqref{eq:dimension-deficit} instead approaches vacuum.

\begin{figure}[!htbp]
\centering\includegraphics[width=\linewidth]{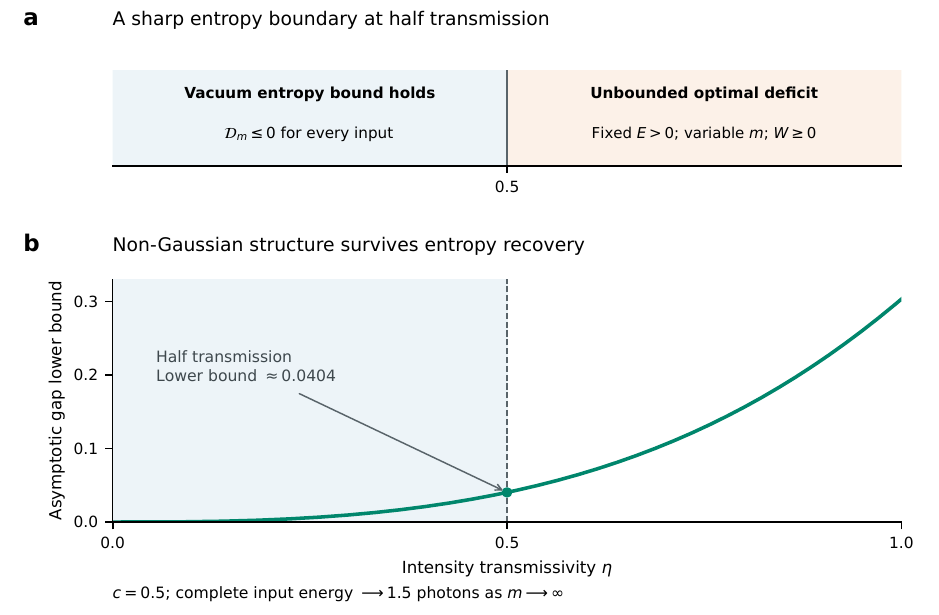}
\caption{\textbf{Entropy recovery and structural survival have different loss laws.}
\textbf{a}, Equation~\eqref{eq:main-loss}: for every input budget $E>0$, the maximal Shannon deficit over all finite mode numbers is zero at $\eta\le1/2$ and unbounded at each fixed $\eta>1/2$.
\textbf{b}, The asymptotic Gaussian-hull gap lower bound in Eq.~\eqref{eq:main-hull-gap}, for the separate source with $c=0.5$ and total input energy tending to $1.5$. At half transmission the bound is approximately $0.0404$. The curve is an asymptotic lower bound. Aligned loss axes show structural survival within the region of universal entropy recovery. Appendix~\ref{app:residual} gives the source and witness proof from Refs.~\cite{HeStability,HeConcentration}.}
\label{fig:loss}
\end{figure}

The residual structure also has a processing cost. Let a Gaussian protocol act on one complete copy of $\widetilde\rho_{m,\eta}$ and retain $k$ modes. The protocol may use Gaussian ancillary states and channels, adaptive Gaussian measurements, arbitrary measurable classical feedback, randomisation, discarding and postselection. At actual success probability at least a fixed $0<p_0\le1$, its normalised successful output satisfies
\begin{equation}
d_{\mathcal F_k}(\rho_{\rm out\mid succ})
\le C_{c,p_0}\left(\sqrt{\frac{k}{m}}+m^{-1/4}\right),
\qquad 1\le k\le m,
\label{eq:main-retention}
\end{equation}
for all sufficiently large $m$ \cite{HeConcentration,HeStability}. The constant is uniform over protocols; the comparison set includes Gaussian mixtures of unrestricted energy. Appendix~\ref{app:retention-interface} gives the precise interface and its transfer through loss. Consequently, $k=o(m)$ makes the output gap vanish. Retaining any fixed gap smaller than the lower bound in Eq.~\eqref{eq:main-hull-gap} requires $k$ proportional to $m$, and keeping the full state attains this order. The minimum retained mode count is therefore $\Theta(m)$ at fixed $p_0$.

\section{From counterexamples to stability questions}
\label{sec:discussion}

The counterexamples distinguish three notions that vacuum intuition can conflate. Positive Wigner density is a condition on an entire phase space. Small mean energy is a moment condition. Small distance from a Gaussian state is an operational closeness condition. Equations~\eqref{eq:central-balance} and \eqref{eq:dimension-deficit} show why these constraints can coexist with entropy below vacuum, and even with an unbounded total deficit in growing dimension.

Earlier work on complex Wigner entropy studied fully coherent binary Fock superpositions and their incoherent endpoints \cite{CHV2024}. The family in Eq.~\eqref{eq:main-doublet} resolves the intermediate-coherence window where the Wigner function is positive and its ordinary Shannon entropy is already below vacuum. The fixed-$\lambda$ one-photon coefficient agrees with the positive Fock-qubit result \cite{QianGagatsos2024}; the sixth-order construction uses a thermal component of infinite Fock support. The phase-averaged squeezed family in the sharp lower bound was previously considered in another role \cite{MKC2008v1}; its energy-dependent parameters and the uniform displacement estimate establish the matching entropy scale here.

The companion theory resolves further parts of the stability landscape \cite{HeStability}. Near R\'enyi order two, $\alpha=2-\delta$, and at fixed positive total energy, the first mode count attaining a fixed positive deficit is
$\Theta(\delta^{-1}\ln(1/\delta))$.
At fixed mode number the optimal deficit is of order $\delta/\ln(1/\delta)$; beyond a sufficiently large constant multiple of the crossover scale it is of order $m\delta$. The sharp crossover coefficient remains a joint physical optimisation problem. Under the additional uniform bound $\Tr(1024^{N_m}\rho)\le K$, the same work proves vacuum $L^\alpha$-norm recovery at $\eta=\alpha/2$ in a mode-independent interval near $\alpha=2$. A further exact result covers the full total-photon sector $N_m\le1$ at $\eta=\alpha/2$ for every $1<\alpha<2$, including all coherences and mixtures. The same companion identifies the remaining crossover loss through a joint probability experiment constructed from one physical state and one Gaussian frame. The unrestricted-input R\'enyi boundary and the sharp full-entropy crossover coefficient remain open.

Gaussian rigidity and certification address complementary uses of these structures. Symmetry controls which leading non-Gaussian sectors survive a Gaussian change of frame \cite{HeRigidity}. Certification of a separate low-energy source has a confidence-dependent hierarchy between memory-free local access, local quantum memory and spatially coherent measurements \cite{HeCertification}. These results connect entropy stability to state structure and measurement cost through their stated source families and operations.

A finite-source loss scan in Ref.~\cite{HeStability} provides a direct experimental-theory target: certify an entropy advantage at high transmission and a positive non-Gaussian witness after the Shannon bound has recovered. That model includes source and detector qualification, systematic errors, failed heralds and complete raw-trial budgets. The exceptionally small thermal weight in Fig.~\ref{fig:thermal} serves a different purpose: it isolates the positivity mechanism with a rigorous all-plane certificate.

The central result is a quantitative failure of vacuum entropy minimality. Coherence can lower entropy while preserving positivity; the optimal one-mode decrease vanishes sublinearly with energy; and a growing number of modes turns that recovery into an instability. Half loss restores the universal Shannon bound. The non-Gaussian structure that survives it is controlled by a separate operational theory.

\enlargethispage{2\baselineskip}
\section*{Acknowledgements}
I thank Professors Stephen M. Barnett, Daniel Mulvihill and Joerg Goette for their support of this work. I also thank Professor Barnett for suggesting the independent Wigner-function reconstruction using Eqs.~(4.5.19)--(4.5.21) of Ref.~\cite{BarnettRadmore}.

\FloatBarrier
\clearpage
\appendix
\renewcommand{\thefigure}{S\arabic{figure}}
\setcounter{figure}{0}
\section{A rationally certified thermal counterexample}\label{app:fixed}
For the fixed counterexample, set
\begin{equation}
t=\frac{\sqrt3}{2000},
\qquad
a=t^2=\frac{3}{4{,}000{,}000},
\qquad
\eps=10^{-12},
\qquad
c=\frac{1-\eps}{1+a},
\label{eq:parameters}
\end{equation}
and define
\begin{equation}
|\psi\rangle=\frac{|0\rangle+t|3\rangle}{\sqrt{1+a}},
\qquad
\tau=\frac23\sum_{n=0}^{\infty}\left(\frac13\right)^n|n\rangle\langle n|,
\qquad
\rho=(1-\eps)|\psi\rangle\langle\psi|+\eps\tau.
\label{eq:state}
\end{equation}
Both terms are density operators, hence $\rho\ge0$ and $\Tr\rho=1$. With $\hat N=(\hat q^2+\hat p^2-1)/2$,
\begin{equation}
\langle\hat N\rangle_\rho=3ac+\frac{\eps}{2}<\infty.
\label{eq:energy}
\end{equation}

For this fixed state, the result is
\begin{equation}
W_\rho(q,p)>\frac{1}{4\pi}e^{-(q^2+p^2)}>0
\qquad\text{for every }(q,p)\in\mathbb R^2,
\label{eq:positive-bound}
\end{equation}
and
\begin{equation}
h(W_\rho)-(1+\ln\pi)<0.
\label{eq:entropy-violation}
\end{equation}
This state violates the vacuum Shannon bound.

\subsection{Exact Wigner function}

Write
\begin{equation}
\alpha=\frac{q+ip}{\sqrt2},
\qquad
R=|\alpha|^2,
\qquad
x=q^2+p^2=2R.
\end{equation}
The phase-space density used by Barnett and Radmore is normalized with respect to $\dd^2\alpha$, while Eq.~\eqref{eq:wigner-def} is normalized with respect to $\dd q\,\dd p$; since $\dd^2\alpha=\dd q\,\dd p/2$, the two are related by $W_\rho=W_\alpha/2$.

Let $\Pi=(-1)^{\hat N}$ and $D(\alpha)=\exp(\alpha\hat a^\dagger-\alpha^*\hat a)$. Equation~(4.5.19) of Barnett and Radmore \cite{BarnettRadmore} gives
\begin{equation}
W_\alpha(\alpha)
=\frac{2}{\pi}\sum_{n\ge0}(-1)^n\langle n|D^\dagger(\alpha)\rho D(\alpha)|n\rangle
=\frac{2}{\pi}\Tr[\rho D(2\alpha)\Pi].
\label{eq:parity}
\end{equation}
Using the normal ordering of $D(2\alpha)$, the required $\{0,3\}$ matrix elements are
\begin{equation}
e^{2R}
\bigl(\langle j|D(2\alpha)\Pi|k\rangle\bigr)_{j,k\in\{0,3\}}
=
\begin{pmatrix}
1 & (2\alpha^*)^3/\sqrt6\\[2mm]
(2\alpha)^3/\sqrt6 & -L_3(4R)
\end{pmatrix}.
\label{eq:matrix-elements}
\end{equation}
The off-diagonal entries are complex conjugates. For the $\tau$ term, the Laguerre generating function yields
\begin{equation}
\frac23e^{-2R}\sum_{n\ge0}\left(-\frac13\right)^nL_n(4R)=\frac12e^{-R}.
\label{eq:tau-sum}
\end{equation}
Substitution into Eq.~\eqref{eq:parity} gives
\begin{equation}
W_\alpha(\alpha)
=\frac{2c}{\pi}e^{-2R}
\left[1-aL_3(4R)+\frac{8t}{\sqrt6}(\alpha^3+\alpha^{*3})\right]
+\frac{\eps}{\pi}e^{-R}.
\label{eq:walpha}
\end{equation}
Hence, in the $(q,p)$ convention,
\begin{equation}
W_\rho(q,p)
=\frac{e^{-x}}{\pi}
\left\{
c\left[1+aF(x)+\frac{q^3-3qp^2}{500}\right]
+\frac{\eps}{2}e^{x/2}
\right\},
\label{eq:wigner-final}
\end{equation}
where
\begin{equation}
F(x)=-L_3(2x)=-1+6x-6x^2+\frac43x^3.
\label{eq:F}
\end{equation}
Equivalently, with $q+ip=\sqrt{x}e^{i\theta}$,
\begin{equation}
W_\rho=\pi^{-1}e^{-x}g(x,\theta),
\qquad
g=A(x)+B(x)\cos3\theta,
\label{eq:gdef}
\end{equation}
where
\begin{equation}
A(x)=c[1+aF(x)]+\frac{\eps}{2}e^{x/2},
\qquad
B(x)=\frac{c}{500}x^{3/2}.
\label{eq:AB}
\end{equation}
Direct angular averaging gives
\begin{equation}
\int_{\mathbb R^2}W_\rho\,\dd q\,\dd p=1,
\qquad
\mathbb E_\rho[x]=1+6ac+\eps=1+2\langle\hat N\rangle_\rho.
\label{eq:normalization-moment}
\end{equation}
Appendix~\ref{app:second} reproduces Eq.~\eqref{eq:wigner-final} by a separate calculation from Eq.~(4.5.21) of Ref.~\cite{BarnettRadmore}.

\subsection{Global positivity}

Since $B(x)\ge0$, angular minimization gives $g\ge A-B$. Put $y=\sqrt{x}$. Replacing the positive exponential in $A$ by its Taylor polynomial through order $16$ in $x/2$ gives the global lower bound
\begin{equation}
\begin{split}
g(x,\theta)\ge P(y):={}&c\left[1-a+6ay^2-\frac{y^3}{500}-6ay^4+\frac{4a}{3}y^6\right]\\
&+\frac{\eps}{2}\sum_{k=0}^{16}\frac{y^{2k}}{2^k k!}.
\end{split}
\label{eq:Py}
\end{equation}
This is a degree-$32$ polynomial with rational coefficients. On each interval
\begin{equation}
I_k=[k/2,(k+1)/2],\qquad k=0,\ldots,15,
\end{equation}
write
\begin{equation}
P\!\left(\frac{k}{2}+\frac{z}{2}\right)=\sum_{j=0}^{32}d_jz^j,
\end{equation}
and convert to the degree-$32$ Bernstein basis. The Bernstein coefficients are
\begin{equation}
b_i=\sum_{j=0}^{i}d_j\frac{\binom{i}{j}}{\binom{32}{j}}.
\label{eq:bernstein}
\end{equation}
All quantities in this finite comparison are rational. Exact arithmetic gives
\begin{equation}
\min_{k,i}b_{k,i}>\frac{2747}{10000}>\frac14.
\label{eq:bernstein-bound}
\end{equation}
Since Bernstein basis functions are nonnegative and sum to one, $P>1/4$ on $0\le y\le8$. On the remaining half-line, the constant coefficient of $P(8+z)$ exceeds $1/4$ and every positive-power coefficient is positive. Thus $P>1/4$ also for $y\ge8$. Equations~\eqref{eq:gdef} and \eqref{eq:Py} therefore imply Eq.~\eqref{eq:positive-bound}.

The same estimates give an upper bound $W_\rho\le Ce^{-x/2}$ for a finite constant $C$. Together with Eq.~\eqref{eq:positive-bound} and Eq.~\eqref{eq:normalization-moment}, this implies
\begin{equation}
\int_{\mathbb R^2}W_\rho|\ln W_\rho|\,\dd q\,\dd p<\infty.
\label{eq:entropy-finite}
\end{equation}

\subsection{A finite rational entropy certificate}

For $A$ and $B$ in Eq.~\eqref{eq:AB}, define
\begin{equation}
d(x)=\sqrt{A(x)^2-B(x)^2},
\qquad
K(x)=\frac{A(x)+d(x)}{2},
\qquad
r(x)=\frac{B(x)}{A(x)+d(x)}.
\label{eq:dKr}
\end{equation}
Strict positivity implies $|r|<1$. The factorization
\begin{equation}
g=K(1+re^{3i\theta})(1+re^{-3i\theta})
\end{equation}
gives
\begin{equation}
\langle\ln g\rangle_\theta=\ln K,
\qquad
\langle\cos(3\theta)\ln g\rangle_\theta=r,
\end{equation}
and therefore
\begin{equation}
J(x):=\langle g\ln g\rangle_\theta
=A\ln\!\left(\frac{A+d}{2}\right)+A-d.
\label{eq:J}
\end{equation}
Using $\ln W_\rho=-\ln\pi-x+\ln g$ and Eq.~\eqref{eq:normalization-moment},
\begin{equation}
\Delta:=h(W_\rho)-(1+\ln\pi)
=6ac+\eps-\int_0^\infty e^{-x}J(x)\,\dd x.
\label{eq:delta}
\end{equation}

The strict sign can be certified with finite rational arithmetic. For every
$g>0$, Taylor's theorem and $(g\ln g)^{(4)}=2/g^3>0$ imply
\begin{equation}
 g\ln g\ge u+\frac{u^2}{2}-\frac{u^3}{6},\qquad u=g-1.
 \label{eq:thermal-cubic-certificate}
\end{equation}
Put $u_0=A-1$. Exact angular averaging gives
\[
 J(x)\ge P(x):=u_0+\frac{u_0^2}{2}-\frac{u_0^3}{6}
                   +\frac{B^2}{4}(1-u_0).
\]
On $[0,32]$, expand the displayed expression using
$u_0=c[1+aF(x)]-1+(\eps/2)e^{x/2}$ and
$B^2=c^2x^3/250000$. It is a finite polynomial in $x$ and $e^{x/2}$
with rational coefficients. After multiplication by $e^{-x}$, its
integrals are rational combinations of
\[
 I_n(s)=\int_0^{32}x^n e^{sx}\dd x,
 \qquad s\in\{-1,-1/2,0,1/2\},\quad 0\le n\le9.
\]
They are evaluated by the finite recurrence
\[
 I_0(s)=\frac{e^{32s}-1}{s},\qquad
 I_n(s)=\frac{32^ne^{32s}}s-\frac ns I_{n-1}(s)
 \quad(s\ne0),\qquad I_n(0)=\frac{32^{n+1}}{n+1}.
\]
All exponentials can be enclosed using a single rational interval. With
$S_{128}=\sum_{j=0}^{128}16^j/j!$,
\[
 S_{128}\le e^{16}\le S_{128}
       +\frac{16^{129}/129!}{1-16/130}.
\]
The upper bound is the geometric bound on the positive Taylor tail.
Take reciprocals and squares to enclose $e^{-16}$ and $e^{-32}$, and
propagate these rational intervals through the finite recurrence.
Finally, $g\ln g\ge-1/e$ gives
$\int_{32}^\infty e^{-x}J(x)\dd x\ge-e^{-32}$.
These steps give the fully rational comparison
\begin{equation}
 \boxed{\displaystyle
 \Delta\le6ac+\eps-\int_0^{32}e^{-x}P(x)\dd x+e^{-32}
 <-\frac{149}{10^8}<0.}
 \label{eq:delta-final}
\end{equation}
For reference, the upper expression is approximately
$-1.49978297\times10^{-6}$. Direct numerical integration of the exact
angular formula gives $\Delta\simeq-1.50014304\times10^{-6}$.
The strict sign follows from the rational comparison above.

\section{Finite-support counterexamples}
\label{sec:finite-support}

We next restrict the state to the two-dimensional subspace spanned by $|0\rangle$ and $|n\rangle$. Within this subspace, the off-diagonal coherence can be reduced independently of the two populations.

For an integer $n\ge1$, a real amplitude $t$, and a coherence fraction $0\le\lambda<1$, define
\begin{equation}
\sigma_{n,t,\lambda}=\frac{|0\rangle\langle0|+t^2|n\rangle\langle n|+\lambda t(|0\rangle\langle n|+|n\rangle\langle0|)}{1+t^2}.
\label{eq:doublet-state}
\end{equation}
The nonzero $2\times2$ block has determinant $t^2(1-\lambda^2)/(1+t^2)^2$, so $\sigma_{n,t,\lambda}$ is a physical rank-two state for $t\ne0$ and $\lambda<1$. Its mean photon number and impurity are
\begin{equation}
E_{n,t}:=\Tr(\hat N\sigma_{n,t,\lambda})=\frac{nt^2}{1+t^2},\qquad
1-\Tr\sigma_{n,t,\lambda}^2=\frac{2(1-\lambda^2)t^2}{(1+t^2)^2}.
\label{eq:doublet-energy}
\end{equation}
Let
\begin{equation}
F_n(x)=(-1)^nL_n(2x),\qquad Q_n(x)=\frac{(2x)^{n/2}}{\sqrt{n!}}.
\label{eq:FnQn}
\end{equation}
The parity representation used in Appendix~\ref{app:fixed} gives
\begin{equation}
W_{\sigma_{n,t,\lambda}}=W_0\,g_{n,t,\lambda},\qquad
g_{n,t,\lambda}(x,\theta)=\frac{1+t^2F_n(x)+2\lambda tQ_n(x)\cos(n\theta)}{1+t^2}.
\label{eq:doublet-wigner}
\end{equation}
The finite Laguerre expansion and $\int_0^\infty e^{-x}x^k\,\dd x=k!$ give
\begin{equation}
\int_0^\infty e^{-x}F_n(x)\,\dd x=1,\qquad
\int_0^\infty xe^{-x}F_n(x)\,\dd x=1+2n,
\label{eq:doublet-moments1}
\end{equation}
\begin{equation}
\int_0^\infty e^{-x}Q_n(x)^2\,\dd x=2^n.
\label{eq:doublet-moments2}
\end{equation}
Whenever $g_{n,t,\lambda}>0$, the entropy difference is therefore
\begin{equation}
\Delta_{n,t,\lambda}:=h(W_{\sigma_{n,t,\lambda}})-(1+\ln\pi)
=\frac{2nt^2}{1+t^2}-\int W_0g_{n,t,\lambda}\ln g_{n,t,\lambda}\,\dd q\,\dd p.
\label{eq:doublet-delta}
\end{equation}

\subsection{An explicit vacuum--three-photon family}

\begin{theorem}
\label{thm:explicit-doublet}
For $n=3$, $\lambda=9/10$, and $0<|t|\le1/200$,
\begin{equation}
W_{\sigma_{3,t,9/10}}(q,p)\ge
\frac{1/10-1801t^2}{1+t^2}W_0(q,p)>0
\label{eq:explicit-positive}
\end{equation}
for every $(q,p)$, and
\begin{equation}
h(W_{\sigma_{3,t,9/10}})-(1+\ln\pi)<-\frac{47}{100}t^2.
\label{eq:explicit-negative}
\end{equation}
\end{theorem}

\begin{proof}
Put
\begin{equation}
F(x)=-1+6x-6x^2+\frac43x^3,\qquad Q(x)=\frac2{\sqrt3}x^{3/2}.
\end{equation}
Angular minimization followed by square completion, with $\kappa=9/10$, gives
\begin{align}
1+t^2F-2\lambda|t|Q
&=\left(\sqrt\kappa|t|Q-\frac{\lambda}{\sqrt\kappa}\right)^2
+1-\frac{\lambda^2}{\kappa}+t^2(F-\kappa Q^2),\\
F-\kappa Q^2&=-1+6x-6x^2+\frac2{15}x^3\ge-1801.
\end{align}
The last inequality follows because $(2/15)x^3-6x^2\ge-1800$ on $x\ge0$, with equality at $x=30$, while $6x\ge0$. Since $\lambda=\kappa=9/10$, the numerator of Eq.~\eqref{eq:doublet-wigner} is at least $1/10-1801t^2$, proving Eq.~\eqref{eq:explicit-positive} throughout the stated interval.

For the entropy, use the elementary inequality
\begin{equation}
(1+u)\ln(1+u)\ge u+\frac{u^2}{2}-\frac{u^3}{6},\qquad u>-1.
\label{eq:scalar-log-bound}
\end{equation}
Its remainder is nonnegative because the fourth derivative of $(1+u)\ln(1+u)$ is $2/(1+u)^3$. Let $a=t^2$ and $H=F-1$. Then
\begin{equation}
u:=g_{3,t,\lambda}-1=\frac{aH+2\lambda tQ\cos3\theta}{1+a},\qquad \int W_0u=0.
\end{equation}
Direct polynomial integration gives
\begin{equation}
\int W_0H^2=244,\qquad \int W_0H^3=173392,\qquad
\int W_0Q^2=8,\qquad \int W_0HQ^2=496.
\label{eq:explicit-moments}
\end{equation}
Using $\langle\cos3\theta\rangle=\langle\cos^33\theta\rangle=0$ and $\langle\cos^23\theta\rangle=1/2$,
\begin{align}
\int W_0u^2&=\frac{16\lambda^2a+244a^2}{(1+a)^2},\\
\int W_0u^3&=\frac{2976\lambda^2a^2+173392a^3}{(1+a)^3}.
\end{align}
Equations~\eqref{eq:doublet-delta} and \eqref{eq:scalar-log-bound}, with $\lambda=9/10$, yield
\begin{equation}
\Delta_{3,t,9/10}\le
\frac{4a(539675a^2+5349a-9)}{75(1+a)^3}.
\label{eq:explicit-rational-upper}
\end{equation}
To compare this with $-(47/100)a$, it is enough to check
\begin{equation}
-\frac1{100}+\frac{28669}{100}a+\frac{8635223}{300}a^2+\frac{47}{100}a^3<0.
\end{equation}
The left side is increasing for $a\ge0$ and is negative at the upper endpoint $a=1/40000$. This proves Eq.~\eqref{eq:explicit-negative}.
\end{proof}

For example, $t=1/1000$ gives the exact lower bound
\begin{equation}
W_{\sigma_{3,1/1000,9/10}}\ge\frac{98199}{1000001}W_0>0,
\end{equation}
and Eq.~\eqref{eq:explicit-rational-upper} gives
\begin{equation}
\Delta_{3,1/1000,9/10}< -4.79\times10^{-7}.
\end{equation}
\subsection{Entropy expansion at fixed Fock level}

\begin{theorem}
\label{thm:general-doublet}
Fix an integer $n\ge1$ and $0\le\lambda<1$. There is a $t_0=t_0(n,\lambda)>0$ such that $W_{\sigma_{n,t,\lambda}}>0$ everywhere for $|t|\le t_0$. Moreover,
\begin{equation}
\Delta_{n,t,\lambda}=(2n-2^n\lambda^2)t^2+O_{n,\lambda}(t^4),\qquad t\to0.
\label{eq:general-doublet-expansion}
\end{equation}
Consequently, for every $n\ge3$ and
\begin{equation}
\sqrt{\frac{n}{2^{n-1}}}<\lambda<1,
\label{eq:coherence-threshold}
\end{equation}
all sufficiently small nonzero $t$ give strictly Wigner-positive finite-support counterexamples.
\end{theorem}

\begin{proof}
Choose $\lambda^2<\kappa<1$. Since $F_n$ and $Q_n^2$ have the same positive leading coefficient, $F_n-\kappa Q_n^2$ has a finite lower bound $-C$ on $[0,\infty)$. Hence
\begin{align}
1+t^2F_n-2\lambda|t|Q_n
&=\left(\sqrt\kappa|t|Q_n-\frac{\lambda}{\sqrt\kappa}\right)^2
+1-\frac{\lambda^2}{\kappa}+t^2(F_n-\kappa Q_n^2)\\
&\ge1-\frac{\lambda^2}{\kappa}-Ct^2.
\end{align}
For sufficiently small $|t|$, this is strictly positive uniformly in $x$ and $\theta$.

On a fixed closed $t$-interval containing zero, $g_{n,t,\lambda}\ge m>0$. Equation~\eqref{eq:doublet-wigner} bounds $g_{n,t,\lambda}$ and its first four $t$-derivatives by $K(1+x^n)$. Thus $Ke^{-x}(1+x^{4n})[1+\ln(2+x)]$ bounds the derivatives of $e^{-x}g_{n,t,\lambda}\ln g_{n,t,\lambda}$ through order four. This integrable bound justifies four differentiations under the integral. Since $\theta\mapsto\theta+\pi/n$ takes $t$ to $-t$, the entropy is even. At $t=0$,
\begin{equation}
g=1,\qquad \partial_tg=2\lambda Q_n\cos(n\theta),\qquad
\int W_0\,\partial_t^2g=0.
\end{equation}
Using Eq.~\eqref{eq:doublet-moments2}, the second-order coefficient of $\int W_0g\ln g$ is
\begin{equation}
\frac12\int W_0(\partial_tg)^2=2^n\lambda^2.
\end{equation}
The energy term in Eq.~\eqref{eq:doublet-delta} contributes $2nt^2+O_n(t^4)$, proving Eq.~\eqref{eq:general-doublet-expansion}.
\end{proof}

For $n=1,2$, the coefficient in Eq.~\eqref{eq:general-doublet-expansion} is strictly positive for every fixed $\lambda<1$. Within this fixed-$\lambda$ two-level quadratic mechanism, a negative quadratic term first becomes possible at $n=3$, with the coherence threshold in Eq.~\eqref{eq:coherence-threshold}.

\subsection{The joint positivity and entropy window}
\label{app:window}

For $n=3$, write $a=t^2$ and $Q(x)=2x^{3/2}/\sqrt3$. The angular minimum is
\[
\frac{1+aF(x)-2\lambda tQ(x)}{1+a}.
\]
Hence, for $t>0$, its nonnegativity is equivalent to
$\lambda\le\lambda_{\rm pos}(t)$ as defined in Eq.~\eqref{eq:main-window}.
Differentiating the function being minimised shows that its derivative has the sign of
$a(4x^3-6x^2-6x+3)-3$.
For $0<t<1$ this polynomial is negative at zero, decreases until
$x=(1+\sqrt3)/2$ and then increases strictly to infinity. The minimised function tends to infinity at both ends, so its unique positive stationary point is its global minimum.

The same second and third moments as in Eqs.~\eqref{eq:explicit-moments}, without fixing $\lambda$, give
\[
\frac{\Delta_{3,t,\lambda}}{t^2}\le
\frac{6-8\lambda^2+(-110+488\lambda^2)a+(86348/3)a^2}{(1+a)^3}.
\]
For the amplitudes in Fig.~\ref{fig:window}, $8-488a>0$. Solving for a negative numerator gives Eq.~\eqref{eq:main-window-cert}. This is an analytic sufficient entropy condition on the same state whose positivity was just established.

There is also a unique entropy crossing whenever the two signs occur. At fixed $t$, write $g=A+\lambda C$, where $A=(1+aF)/(1+a)$ and $C=2tQ\cos3\theta/(1+a)$. In the strict positivity interval,
\[
\frac{\mathrm d^2}{\mathrm d\lambda^2}\KL(W_0g\Vert W_0)
=\int W_0\frac{C^2}{A+\lambda C}>0.
\]
The first derivative vanishes at $\lambda=0$ by angular symmetry. Consequently $\KL$ increases and $\Delta$ decreases strictly for $\lambda>0$.
At $t=0.005$, numerical evaluations give
$\lambda_{\rm pos}\simeq0.92959375$,
$\lambda_{\rm cert}\simeq0.86648917$ and an entropy zero near $0.86583712$.

The numerical contour uses the exact angular integral
\[
\langle g\ln g\rangle_\theta
=A\ln\frac{A+\sqrt{A^2-B^2}}2+A-\sqrt{A^2-B^2},
\qquad B=\frac{4\lambda t x^{3/2}}{\sqrt3(1+a)},
\]
followed by radial quadrature against $e^{-x}\dd x$. Composite Gauss quadrature was compared with 50-digit quadrature at six interior points; the largest absolute entropy difference was below $9\times10^{-17}$. Four crossing evaluations at 60 digits had residuals below $9\times10^{-17}$. On the displayed parameter range, an absolute bound on the omitted radial integral beyond $x=100$ is $6.4\times10^{-35}$. These numerical checks locate the contour; the polynomial inequality certifies the hatched region.

\begin{figure}[!htbp]
\centering\includegraphics[width=\linewidth]{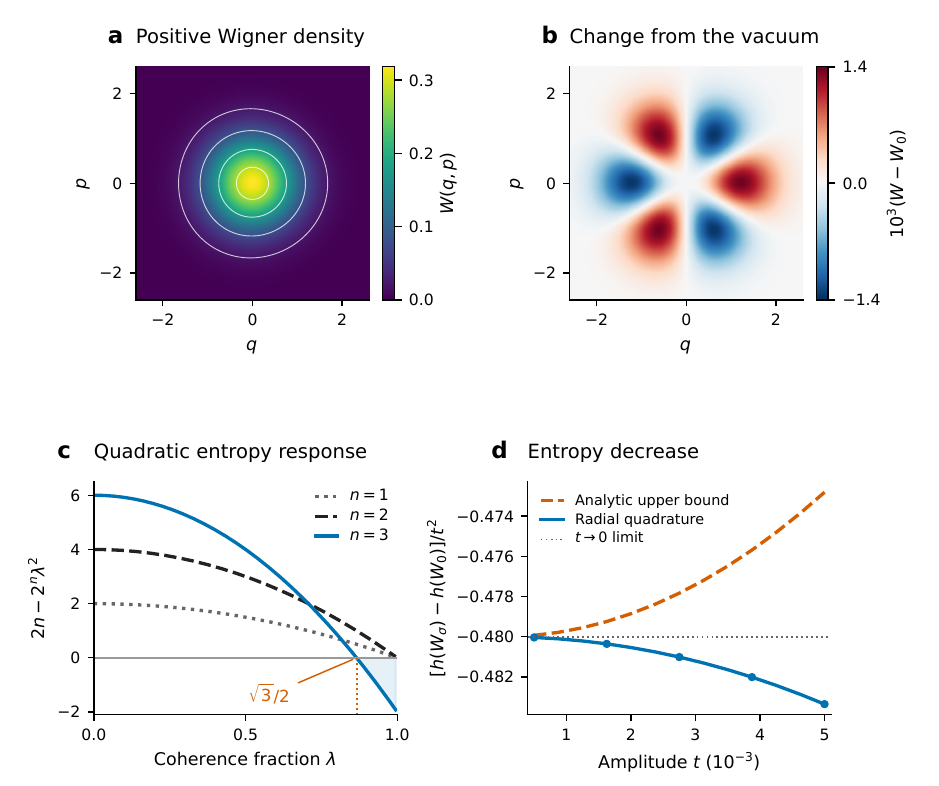}
\caption{\textbf{Phase-space structure and the fixed-coherence response.}
\textbf{a,b}, The Wigner density and its difference from vacuum for $\sigma_{3,1/200,9/10}$. The difference colour scale is multiplied by $10^3$. Equation~\eqref{eq:explicit-positive} gives $W_\sigma/W_0>0.05497$ globally.
\textbf{c}, The coefficient $2n-2^n\lambda^2$ from Theorem~\ref{thm:general-doublet}; its first negative quadratic direction occurs at $n=3$.
\textbf{d}, At $\lambda=0.9$, the exact angular entropy integral followed by radial quadrature gives the solid curve, with markers selecting points. The dashed curve is the analytic upper bound~\eqref{eq:explicit-rational-upper}, divided by $t^2$; the dotted line is its limit $-0.48$.}
\label{fig:mechanism}
\end{figure}

\section{Minimum thermal weight for Wigner positivity}

Return to the coherent superposition in Eq.~\eqref{eq:state}, now with variable amplitude $t>0$, and keep the thermal state $\tau$ fixed. We determine the smallest mixing weight for which the Wigner function is nonnegative.

Let
\begin{equation}
|\psi_t\rangle=\frac{|0\rangle+t|3\rangle}{\sqrt{1+t^2}},\qquad
\tau=\frac23\sum_{j=0}^\infty3^{-j}|j\rangle\langle j|,
\end{equation}
and
\begin{equation}
\rho^{\rm th}_{t,\epsilon}=(1-\epsilon)|\psi_t\rangle\langle\psi_t|+\epsilon\tau,
\qquad W_\tau=\frac1{2\pi}e^{-x/2}.
\end{equation}
Let $\epsilon_*(t)$ be the least $\epsilon\in[0,1]$ for which $W_{\rho^{\rm th}_{t,\epsilon}}\ge0$ everywhere.

\begin{proposition}
\label{prop:repair}
For all sufficiently small $t>0$, define
\begin{equation}
M(t)=\sup_{x\ge0}\frac{2e^{-x/2}}{1+t^2}
\left[-1-t^2F(x)+\frac{4t}{\sqrt3}x^{3/2}\right]_+.
\label{eq:Mrepair}
\end{equation}
Here $[u]_+=\max\{u,0\}$. Then
\begin{equation}
\epsilon_*(t)=\frac{M(t)}{1+M(t)},
\label{eq:epsstar}
\end{equation}
and
\begin{equation}
\lim_{t\downarrow0}t^{2/3}\ln\epsilon_*(t)
=-\frac12\left(\frac34\right)^{1/3}.
\label{eq:repair-asymptotic}
\end{equation}
At $\epsilon=\epsilon_*(t)$, the Wigner function has a zero. It is strictly positive for $\epsilon\in(\epsilon_*(t),1]$.
\end{proposition}

\begin{proof}
Put $\xi=2t/\sqrt3$ and $a=t^2=3\xi^2/4$. The sign of the angular minimum is determined by
\begin{equation}
P_\xi(x)=1-2\xi x^{3/2}+aF(x)
=(1-\xi x^{3/2})^2-a+6ax-6ax^2.
\end{equation}
At $x_\xi=\xi^{-2/3}$, $P_\xi(x_\xi)=-a(1-6x_\xi+6x_\xi^2)<0$ for sufficiently small $\xi$. The function in Eq.~\eqref{eq:Mrepair} is continuous, tends to zero at infinity, and is positive at $x_\xi$. Its maximum $M(t)>0$ is therefore attained at a finite $x_*$. Global Wigner nonnegativity is equivalent to $\epsilon-(1-\epsilon)M(t)\ge0$, which gives Eq.~\eqref{eq:epsstar}. Equality produces a zero at $x_*$ and a minimizing angle; a larger mixing weight in $[0,1]$ gives strict positivity.

The test point gives
\begin{equation}
M(t)\ge\frac{2a}{1+a}(1-6\xi^{-2/3}+6\xi^{-4/3})
\exp[-1/(2\xi^{2/3})],
\end{equation}
whose prefactor is $9\xi^{2/3}(1+o(1))$. Hence
\begin{equation}
\liminf_{\xi\downarrow0}\xi^{2/3}\ln M\ge-\frac12.
\end{equation}
Conversely, fix $b\in(0,1)$ and put $x_b=b\xi^{-2/3}$. For sufficiently small $\xi$, $P_\xi>0$ on $[0,x_b]$. Indeed, for $x\le b\xi^{-2/3}$, the square term is at least $(1-b^{3/2})^2$, while the remaining terms are $O_b(\xi^{2/3})$ uniformly. On $x\ge x_b$, the polynomial-Gaussian tail gives
\begin{equation}
M(t)\le\frac{2}{1+a}\left(2b^{3/2}+a+\frac92b^2\xi^{2/3}\right)
\exp[-b/(2\xi^{2/3})].
\end{equation}
Thus $\limsup\xi^{2/3}\ln M\le-b/2$; taking $b\uparrow1$ gives $\xi^{2/3}\ln M\to-1/2$. Since $M\to0$, Eq.~\eqref{eq:epsstar} has the same logarithmic asymptotic. Converting from $\xi$ to $t$ yields Eq.~\eqref{eq:repair-asymptotic}. Equivalently,
\[
\epsilon_*(t)=\exp\!\left[-\frac12\left(\frac34\right)^{1/3}t^{-2/3}+o(t^{-2/3})\right]
\qquad (t\downarrow0).
\]
\end{proof}

\subsection{An exact certificate for the smaller thermal repair}
\label{app:tight-repair}

Fix $a=3/4000000$ and $\epsilon_1=2\times10^{-21}$. The angular minimum of $W_{\rho_{\rm f}}/W_0$, with $y=\sqrt{x}$ and $c_1=(1-\epsilon_1)/(1+a)$, is
\[
c_1P(y)+\frac{\epsilon_1}{2}e^{y^2/2},\qquad
P(y)=1-a+6ay^2-\frac{y^3}{500}-6ay^4+\frac{4a}{3}y^6.
\]
The following finite rational comparison proves a global lower bound; its partitions and truncation are given explicitly to make it reproducible.

On each of the 18 intervals $[k/2,(k+1)/2]$, $0\le k\le17$, convert the degree-six polynomial $P$ to its Bernstein basis using Eq.~\eqref{eq:bernstein} with degree six. All coefficients are positive. For $y\ge11$, all coefficients of $P(11+z)$ are positive. These bounds control the two exterior regions.

For the remaining interval $[9,11]$, set $l_j=9+j/128$ and $r_j=l_j+1/128$ for $0\le j\le255$. Let $b_j$ be the smallest degree-six Bernstein coefficient of $P$ on $[l_j,r_j]$. Throughout that interval,
\[
c_1P(y)+\frac{\epsilon_1}{2}e^{y^2/2}
\ge c_1b_j+\frac{\epsilon_1}{2}
\sum_{k=0}^{96}\frac{(l_j^2/2)^k}{k!}.
\]
All quantities on the right are rational. The smallest of these 256 bounds and the two exterior lower bounds exceeds $79/10000$ by exact rational comparison. Its decimal value is approximately $0.00792693364$. This proves the positivity part of Eq.~\eqref{eq:main-new-certified}.

The negative interval of the uncorrected angular minimum is
$87.311619\ldots<x<115.690053\ldots$.
For completeness, exact Sturm counts give two positive roots of $P(y)$, one in $[9.344,9.345]$ and one in $[10.755,10.756]$. The polynomial $P'(y)-yP(y)$ has exactly one root between these brackets. Therefore the demand in Eq.~\eqref{eq:main-repair} has one interior maximum. Its numerical location is $x\simeq89.1900077$, with $M\simeq9.0397820\times10^{-22}$. The exact positivity certificate above is independent of this numerical maximisation.

To prove the entropy decrease, let $\rho_0$ be the certified state in Appendix~\ref{app:fixed}, with the same $t$ and $\epsilon_0=10^{-12}$. Define
\[
r=\frac{\epsilon_0-\epsilon_1}{1-\epsilon_1}\in(0,1).
\]
Then $\rho_0=(1-r)\rho_{\rm f}+r\tau$. Every Wigner density in this identity is nonnegative, and $h(W_\tau)=h_0+\ln2$. Concavity of differential entropy gives
\begin{equation}
h(W_{\rho_{\rm f}})-h_0
\le\frac{\Delta_0-r\ln2}{1-r}<0,
\qquad \Delta_0=h(W_{\rho_0})-h_0<0 .
\label{eq:tight-concavity}
\end{equation}
Equation~\eqref{eq:delta-final} supplies a strict upper bound for $\Delta_0$. Thus the new entropy sign follows analytically from a positive-density mixture, without attributing Shannon entropy to the negative Wigner function of the coherent component.

For Fig.~\ref{fig:thermal}c, use $A=c_1[1+aF(x)]+(\epsilon_1/2)e^{x/2}$, $B=c_1x^{3/2}/500$ and the angular integral $J=\langle(A+B\cos3\theta)\ln(A+B\cos3\theta)\rangle_\theta$. Then
\[
H_{\rho_{\rm f}}(X)-H_0(X)
=\int_0^Xe^{-x}\{(A-1)(\ln\pi+x)-J(x)\}\,\dd x.
\]
The curve uses composite 32-point Gauss quadrature on intervals of length $0.05$ up to $X=32$. Values at $X=1,4,8,16,24,32$ agree with 70-digit quadrature within $3\times10^{-16}$. The complete numerical entropy difference, approximately $-1.5001440445\times10^{-6}$, also agrees when computed through Eq.~\eqref{eq:central-balance}. Its negative sign is established by Eq.~\eqref{eq:tight-concavity}.

\section{Energy response and a one-photon core}
\subsection{No uniform linear energy correction}

The entropy deficit per mean photon number is unbounded arbitrarily close to zero energy.

\begin{corollary}
\label{cor:no-linear}
For every finite $C>0$ and every $E_0>0$, there is a one-mode physical state $\sigma$ with finite Fock support such that
\begin{equation}
W_\sigma>0,\qquad 0<\Tr(\hat N\sigma)<E_0,
\end{equation}
and
\begin{equation}
h(W_\sigma)<1+\ln\pi-C\Tr(\hat N\sigma).
\label{eq:no-linear}
\end{equation}
Hence no finite linear correction in the mean photon number restores a universal vacuum lower bound in any neighborhood of zero energy.
\end{corollary}

\begin{proof}
Fix $\lambda=9/10$. By Theorem~\ref{thm:general-doublet}, for any fixed $n$,
\begin{equation}
\frac{-\Delta_{n,t,\lambda}}{E_{n,t}}
\longrightarrow \frac{2^n\lambda^2}{n}-2
\qquad(t\to0).
\end{equation}
Choose a finite $n$ large enough that the right side exceeds $C$. For this fixed $n$, all sufficiently small nonzero $t$ give strict Wigner positivity and preserve the strict ratio inequality. Taking $t$ smaller if necessary also makes $E_{n,t}<E_0$, proving Eq.~\eqref{eq:no-linear}.
\end{proof}

\subsection{A sixth-order counterexample with a one-photon core}
\label{sec:one-photon}

The positive quadratic coefficient for $n=1$ in Theorem~\ref{thm:general-doublet} assumes a fixed coherence fraction $\lambda<1$. A fully coherent vacuum--one-photon component, with a thermal admixture tending to zero, gives a different small-amplitude limit. Throughout this section, $\tau$ is the thermal state in Eq.~\eqref{eq:state} and $T=W_\tau=(2\pi)^{-1}e^{-x/2}$.

\begin{theorem}
\label{thm:one-photon}
For sufficiently small $t>0$, define
\begin{equation}
 |\chi_t\rangle=\frac{|0\rangle+t|1\rangle}{\sqrt{1+t^2}},\qquad
 \epsilon_t=4t^2\exp\!\left(-\frac{1}{4t^2}+\frac{1}{2t}\right),\qquad
 \omega_t=(1-\epsilon_t)|\chi_t\rangle\langle\chi_t|+\epsilon_t\tau.
 \label{eq:one-state}
\end{equation}
Then $\omega_t$ is a finite-energy physical state with
\begin{equation}
 W_{\omega_t}\ge\frac{\epsilon_t}{2}T>0,\qquad
 E_t:=\Tr(\hat N\omega_t)=\frac{(1-\epsilon_t)t^2}{1+t^2}+\frac{\epsilon_t}{2},
 \label{eq:one-positive-energy}
\end{equation}
and
\begin{equation}
 h(W_{\omega_t})-(1+\ln\pi)=-\frac43t^6+o(t^6).
 \label{eq:one-sixth}
\end{equation}
\end{theorem}

\begin{proof}
The parity formula gives
\begin{equation}
 W_{|\chi_t\rangle\langle\chi_t|}
 =\frac{W_0}{1+t^2}\left[1+2\sqrt2tq+t^2(2x-1)\right].
 \label{eq:one-wigner}
\end{equation}
The bracket equals
\begin{equation}
 2t^2\left[\left(q+\frac{1}{\sqrt2t}\right)^2+p^2\right]-t^2.
 \label{eq:one-disk}
\end{equation}
Its negative part is therefore confined to the disk of radius $1/\sqrt2$ centred at $(-1/(\sqrt2t),0)$. For $t<1$, the smallest $x$ in this disk is $1/(2t^2)-1/t+1/2$. Writing $W_t=W_{|\chi_t\rangle\langle\chi_t|}$, we obtain on this disk
\begin{equation}
 \frac{W_t^-}{T}\le \frac{2t^2}{1+t^2}
 \exp\!\left(-\frac{1}{4t^2}+\frac{1}{2t}-\frac14\right)
 <\frac{\epsilon_t}{2}.
 \label{eq:one-negative-bound}
\end{equation}
Outside it, $W_t\ge0$. Since $\epsilon_t\to0$, the mixture is physical for sufficiently small $t$, and Eq.~\eqref{eq:one-positive-energy} follows on the full plane.

Put $P_t=W_{\omega_t}$ and $D_t=32\ln(1/t)$. On $x\le D_t$, write
\begin{equation}
 \frac{W_t}{W_0}=1+U_t,\qquad
 U_t=\frac{tA+t^2B}{1+t^2},\qquad A=2\sqrt2q,\quad B=2(x-1).
 \label{eq:one-U}
\end{equation}
Here $\sup_{x\le D_t}|U_t|\to0$. The difference between $P_t/W_0$ and $1+U_t$ is bounded by
\begin{equation}
 \epsilon_t\left[C(1+D_t)+\frac12e^{D_t/2}\right],
 \label{eq:one-inner-repair}
\end{equation}
which is smaller than every fixed power of $t$.
For $\Phi(y)=y\ln y-y+1$, Taylor's formula on this disk gives
\begin{equation}
 \Phi(1+u)=\sum_{j=2}^{6}\frac{(-1)^j u^j}{j(j-1)}+O(|u|^7).
 \label{eq:one-Taylor}
\end{equation}
Its integrated remainder is $O(t^7)$: the Gaussian moments of $|A|^7$ and $|B|^7$ are finite. Extending the finite polynomial terms to the plane contributes only $o(t^6)$. With brackets denoting integration against $W_0$, the required moments are
\begin{align}
 \langle A^2\rangle&=4,&\langle B^2\rangle&=4,&
 \langle A^2B\rangle&=8,&\langle A^4\rangle&=48,\notag\\
 \langle B^3\rangle&=16,&\langle A^2B^2\rangle&=48,&
 \langle A^4B\rangle&=192,&\langle A^6\rangle&=960.
 \label{eq:one-moments}
\end{align}
Odd powers of $q$ integrate to zero. Substitution yields
\begin{equation}
 \int_{x\le D_t}W_0\Phi(1+U_t)\,\dd q\,\dd p
 =2t^2-2t^4+\frac{10}{3}t^6+o(t^6).
 \label{eq:one-inner-KL}
\end{equation}

The exterior must be estimated for the repaired, positive density. The lower bound in Eq.~\eqref{eq:one-positive-energy} and $|W_t|\le C(1+x)W_0$ imply
\begin{equation}
 \left|\ln\frac{P_t}{W_0}\right|\le C(1+t^{-2}+x),\qquad
 P_t\le C(1+x)W_0+\epsilon_tT.
 \label{eq:one-log-bound}
\end{equation}
Consequently, for a fixed integer $M$,
\begin{equation}
 \int_{x>D_t}P_t\left|\ln\frac{P_t}{W_0}\right|\,\dd q\,\dd p
 \le Ct^{-2}(1+D_t)^M e^{-D_t}+C\epsilon_t(1+t^{-2})=o(t^6).
 \label{eq:one-tail}
\end{equation}
The exterior mass of $|P_t-W_0|$ is also $o(t^6)$. Normalisation and Eqs.~\eqref{eq:one-inner-repair}--\eqref{eq:one-tail} therefore give the full relative-entropy expansion
\begin{equation}
 D_{\mathrm{KL}}(P_t\|W_0)
 :=\int P_t\ln(P_t/W_0)\,\dd q\,\dd p
 =2t^2-2t^4+\frac{10}{3}t^6+o(t^6).
 \label{eq:one-full-KL}
\end{equation}
Finally,
\begin{equation}
 h(P_t)-(1+\ln\pi)=2E_t-D_{\mathrm{KL}}(P_t\|W_0),\qquad
 2E_t=2t^2-2t^4+2t^6+o(t^6).
 \label{eq:one-energy-series}
\end{equation}
which proves Eq.~\eqref{eq:one-sixth}.
\end{proof}

The $t^2$ and $t^4$ terms in Eqs.~\eqref{eq:one-full-KL} and \eqref{eq:one-energy-series} cancel in the entropy difference; at sixth order their coefficients differ by $2-10/3=-4/3$. In contrast, $\sigma_{1,t,\lambda}$ has the positive leading term $2(1-\lambda^2)t^2$ for every fixed $\lambda<1$. These are distinct paths to the vacuum. The state $\sigma_{1,t,\lambda}$ remains in $\operatorname{span}\{|0\rangle,|1\rangle\}$, whereas $\omega_t$ has a fully coherent one-photon core and the infinite Fock support of $\tau$. The latter is outside the two-level class considered in Ref.~\cite{QianGagatsos2024}.

In terms of the complete state's mean photon number, $E_t\sim t^2$ and
\begin{equation}
 1+\ln\pi-h(W_{\omega_t})=\frac43 E_t^3(1+o(1)).
 \label{eq:one-energy-cubic}
\end{equation}

\section{Sharp low-energy scale}
\label{sec:sharp-energy}\label{app:sharp-statements}

We optimise over the complete class of one-mode Wigner-nonnegative states. Define
\begin{equation}
 \mathcal D(E)=\sup_{\substack{\rho\ge0,\ \Tr\rho=1,\ W_\rho\ge0\\
                           \Tr(\hat N\rho)\le E}}
 \left[1+\ln\pi-h(W_\rho)\right],\qquad L=\ln(1/E).
 \label{eq:optimal-deficit}
\end{equation}
The vacuum is admissible, so the supremum is nonnegative. The parity representation gives $0\le W_\rho\le1/\pi$. For an upper entropy bound, compare $W_\rho$ with the normalised density
\[
 G_E(q,p)=\frac{1}{\pi(2E+1)}\exp\!\left[-\frac{x}{2E+1}\right].
\]
The scalar inequality $s\ln(s/r)-s+r\ge0$ for $s\ge0$ and $r>0$, with the continuous value at $s=0$, gives
\[
 -W_\rho\ln W_\rho\le-W_\rho\ln G_E-W_\rho+G_E.
\]
The right side is integrable because $\int xW_\rho=1+2\Tr(\hat N\rho)\le2E+1$. Integration, together with the parity bound, yields
\begin{equation}
 \ln\pi\le h(W_\rho)\le1+\ln\!\left[\pi(2E+1)\right].
 \label{eq:energy-entropy-finite}
\end{equation}
Thus every entropy in Eq.~\eqref{eq:optimal-deficit} is finite, without a Fock-support cutoff.

Let $r_*\in(1/3,3/8)$ be the unique solution of
\begin{equation}
 r_*\ln\frac{2e}{r_*}=1,
 \label{eq:critical-root}
\end{equation}
and set
\begin{equation}
 b_* =\ln(2/r_*)=r_*^{-1}-1,\qquad
 c_* =\frac{3/2-4r_*/3}{b_*},\qquad
 \gamma=1-r_*\ln2,\qquad \beta=1-c_*\ln2.
 \label{eq:sharp-constants}
\end{equation}
Numerically,
\begin{equation}
 r_*\simeq0.3733646177,\qquad c_*\simeq0.5971235482,\qquad
 \gamma\simeq0.7412033679,\qquad\beta\simeq0.5861054961.
\end{equation}
The constants in Eq.~\eqref{eq:sharp-constants}, rather than these decimal approximations, define the exponents below.

\begin{theorem}
\label{thm:sharp-energy}
There are positive constants $C_1,C_2,E_0$ such that, for every $0<E<E_0$,
\begin{equation}
 C_1\frac{E^\gamma}{[\ln(1/E)]^\beta}
 \le\mathcal D(E)\le
 C_2\frac{E^\gamma}{[\ln(1/E)]^\beta}.
 \label{eq:sharp-energy}
\end{equation}
The lower bound is realised by physical states with strictly positive Wigner functions and mean photon number at most $E$. The upper bound applies to the full class in Eq.~\eqref{eq:optimal-deficit}, including states of infinite Fock support.
\end{theorem}

The proof is given in Appendix~\ref{app:sharp-energy}. Its two estimates use the same energy budget. Wigner positivity bounds the vacuum--Fock coefficients through the angular Fourier coefficients of the excited-state block. A uniform energy-weighted displacement bound enlarges the usable squared radius from $L$ to $L+\frac43\ln L$. The matching construction repairs a vacuum--Fock superposition using a phase-averaged displaced squeezed state, with the squeezing and displacement energies included. The balance occurs at
\begin{equation}
 n=\left\lfloor r_*L+c_*\ln L-K\right\rfloor
 \label{eq:attaining-level}
\end{equation}
for a sufficiently large fixed $K$.

Equation~\eqref{eq:sharp-energy} determines the power and logarithmic correction up to constant factors. In particular,
\begin{equation}
 \mathcal D(E)\longrightarrow0,\qquad
 \frac{\mathcal D(E)}{E}\longrightarrow\infty.
 \label{eq:energy-consequences}
\end{equation}
For the one-photon family, Eqs.~\eqref{eq:one-energy-cubic} and \eqref{eq:sharp-energy} imply
\begin{equation}
 \frac{1+\ln\pi-h(W_{\omega_t})}{\mathcal D(E_t)}
 =O\!\left(E_t^{3-\gamma}[\ln(1/E_t)]^\beta\right)\longrightarrow0.
 \label{eq:one-suboptimal}
\end{equation}
The sixth-order construction identifies a descent direction missed by the quadratic expansion. The order-optimal family in Eq.~\eqref{eq:attaining-state} instead has Fock level growing according to Eq.~\eqref{eq:attaining-level}.

\section{Alternative integral reconstruction}
\label{app:second}

Equation~(4.5.21) of Barnett and Radmore \cite{BarnettRadmore} gives a second calculation directly from the density operator,
\begin{equation}
W_\alpha(\alpha)=\frac{2}{\pi^2}\int_{\mathbb C}\dd^2\beta\,
\langle\alpha+\beta|\rho|\alpha-\beta\rangle
\exp(\alpha^*\beta-\alpha\beta^*).
\label{eq:alt-integral}
\end{equation}
For the weighted pure component,
\begin{equation}
(1-\eps)\langle\alpha+\beta|\psi\rangle\langle\psi|\alpha-\beta\rangle
=c e^{-R-|\beta|^2}
\left[1+\frac{t}{\sqrt6}(\alpha^*+\beta^*)^3\right]
\left[1+\frac{t}{\sqrt6}(\alpha-\beta)^3\right].
\label{eq:pure-alt}
\end{equation}
All moments follow from
\begin{equation}
\begin{aligned}
G(u,v)&=\int_{\mathbb C}e^{-|\beta|^2+u\beta+v\beta^*}\,\dd^2\beta
=\pi e^{uv},\\[1mm]
\mathcal G[f]&:=\frac{e^R}{\pi}\int_{\mathbb C}
 e^{-|\beta|^2+\alpha^*\beta-\alpha\beta^*}f(\beta,\beta^*)\,\dd^2\beta.
\end{aligned}
\label{eq:generator}
\end{equation}
After division by $G(\alpha^*,-\alpha)=\pi e^{-R}$, the required cubic contractions are
\begin{equation}
\mathcal G[(\alpha-\beta)^3]=(2\alpha)^3,
\qquad
\mathcal G[(\alpha^*+\beta^*)^3]=(2\alpha^*)^3,
\label{eq:cubic-contract}
\end{equation}
and
\begin{equation}
\frac16\mathcal G[(\alpha^*+\beta^*)^3(\alpha-\beta)^3]
=\frac16\sum_{j=0}^{3}\binom{3}{j}^2j!(-1)^j(4R)^{3-j}
=-L_3(4R).
\label{eq:product-contract}
\end{equation}
These terms reproduce the pure-state contribution in Eq.~\eqref{eq:walpha}.

For $\tau$,
\begin{equation}
\langle\gamma|\tau|\delta\rangle
=\frac23\exp\left[-\frac{|\gamma|^2+|\delta|^2}{2}+\frac{\gamma^*\delta}{3}\right].
\label{eq:tau-alt}
\end{equation}
Substitution into Eq.~\eqref{eq:alt-integral} gives
\begin{equation}
\frac{4}{3\pi^2}e^{-2R/3}
\int_{\mathbb C}
\exp\left[-\frac43|\beta|^2+\frac23(\alpha^*\beta-\alpha\beta^*)\right]\dd^2\beta
=\frac{1}{\pi}e^{-R}.
\label{eq:tau-integral}
\end{equation}
Thus the integral reproduces Eq.~\eqref{eq:walpha}, and hence Eq.~\eqref{eq:wigner-final}.

\section{Proof of the sharp low-energy bound}
\label{app:sharp-energy}

Throughout this appendix, $E\downarrow0$, $L=\ln(1/E)$, and $r_*,b_*,c_*,\gamma,\beta$ are defined by Eqs.~\eqref{eq:critical-root}--\eqref{eq:sharp-constants}. Constants denoted by $C$ may increase from one estimate to the next, but do not depend on the admissible state. No loss channel is applied. Here $R$ denotes the squared radius in the $(q,p)$ plane, $q^2+p^2=2|\alpha|^2$, rather than $R=|\alpha|^2$ as in Appendix~\ref{app:fixed}.

\subsection{How the sharp exponents arise}\label{app:exponent-guide}

The two exponents follow from the same balance in the universal bound and the attaining construction. For the vacuum--Fock coefficients $v_j=\langle j|\rho|0\rangle$, energy gives $\sum_{j\ge1}j|v_j|^2\le E$. Global Wigner positivity supplies a second constraint at squared radius $R=L+\frac43\ln L$. Splitting the entropy estimate at a Fock level $s$ leaves two leading scales,
\begin{equation}
U_E(s)=\frac{E2^s}{s},\qquad
V_E(s)=\frac{s!}{R^s}.
\label{eq:balance-scales}
\end{equation}
The first controls the low harmonics using their energy cost; the second controls the intermediate harmonics using positivity. Equations~\eqref{eq:Jlow}--\eqref{eq:Jhigh} give the full bounds and show that the higher harmonics and remaining terms are smaller at the chosen split.

To find that split, set $s=\lfloor rL+d\ln L\rfloor$, with fixed $0<r<2/5$. Stirling's formula gives
\begin{align}
\ln U_E(s)&=-(1-r\ln2)L+(d\ln2-1)\ln L+O(1),\notag\\
\ln V_E(s)&=-r(1-\ln r)L+
\left(\frac12+d\ln r-\frac43r\right)\ln L+O(1).
\label{eq:balance-expansions}
\end{align}
Equating the coefficients of $L$ gives $r\ln(2e/r)=1$, hence $r=r_*$. Equating the coefficients of $\ln L$ then gives
\begin{equation}
d\ln(2/r_*)=\frac32-\frac43r_*,\qquad d=c_*.
\label{eq:balance-logarithm}
\end{equation}
Both scales are therefore of order $E^\gamma L^{-\beta}$, with
$\gamma=1-r_*\ln2$ and $\beta=1-c_*\ln2$. The power comes from the exponential competition; the logarithmic correction records the energy weight, Stirling factor and enlarged positivity radius together.

The construction realises this balance with a growing Fock level and a phase-averaged displaced squeezed repair. Its coherence amplitude at the repair radius obeys the exact identity
\begin{equation}
A_R^2=\frac{E}{n}\frac{(2R)^n}{n!}
=\frac{U_E(n)}{V_E(n)},
\label{eq:balance-physical-amplitude}
\end{equation}
using the parameters in Eq.~\eqref{eq:construction-parameters}. Taking $n=\lfloor r_*L+c_*\ln L-K\rfloor$ with a sufficiently large fixed $K$ makes this amplitude small while preserving the scale $E^\gamma L^{-\beta}$. The repair energy is at most $256A_R R^{4/3}e^{-R}=O(A_RE)$, so it fits within the same total budget. Equations~\eqref{eq:total-energy}--\eqref{eq:attaining-positive} prove energy feasibility and positivity, and the final entropy estimate attains the balanced scale.

\subsection{Vacuum decomposition and a local entropy bound}

Write $Q=I-|0\rangle\langle0|$, $\delta=1-\langle0|\rho|0\rangle$, $A=Q\rho Q$, and $v=Q\rho|0\rangle$. For $E<1$, positivity of the block matrix gives
\begin{equation}
 \rho=(1-\delta)|0\rangle\langle0|+|v\rangle\langle0|+|0\rangle\langle v|+A,
 \qquad A\ge0,\quad\Tr A=\delta\le E,
 \label{eq:energy-block}
\end{equation}
\begin{equation}
 \sum_{j\ge1}j|v_j|^2\le(1-\delta)\Tr(\hat N A)\le E.
 \label{eq:coherence-energy}
\end{equation}
Indeed, the Schur complement gives $A\ge |v\rangle\langle v|/(1-\delta)$; taking the trace against truncated photon-number operators and then increasing the truncation proves Eq.~\eqref{eq:coherence-energy}.

With $z=q+ip$, put
\begin{equation}
 f(z)=\sum_{j\ge1}\overline{v_j}\sqrt{\frac{2^j}{j!}}z^j.
 \label{eq:analytic-f}
\end{equation}
The series converges uniformly on compact sets by Cauchy--Schwarz. The vacuum--Fock matrix elements and the parity bound yield
\begin{equation}
 W_\rho=W_0(1-\delta+2\operatorname{Re}f)+W_A,
 \qquad |W_A|\le\frac\delta\pi.
 \label{eq:wigner-block}
\end{equation}
These identities hold for infinite-support states by trace-norm approximation.
Since $W_\rho\ge0$, on $|z|^2\le L$ we have
\begin{equation}
 \operatorname{Re}f(z)\ge-\frac{1-\delta+\delta e^{|z|^2}}2\ge-1.
\end{equation}
The nonnegative real part of $1+f$ has mean one on $|z|=\sqrt L$. Extracting its nonzero Fourier coefficients gives
\begin{equation}
 |v_j|^2\frac{2^j}{j!}\le4L^{-j},\qquad j\ge1.
 \label{eq:old-coefficient}
\end{equation}

Set $U=L-1$ and $g=W_\rho/W_0$. For $\Phi(g)=g\ln g-g+1$, with its continuous value at zero,
\begin{equation}
 0\le\Phi(g)\le(g-1)^2,
\end{equation}
and on the disk $x\le U$,
\begin{equation}
 -W_\rho\ln W_\rho+W_0\ln W_0
 =W_0(g-1)(x+\ln\pi-1)-W_0\Phi(g).
 \label{eq:local-entropy-identity}
\end{equation}
Write $g-1=d+2\operatorname{Re}f$, where $d=W_A/W_0-\delta$. Its nonanalytic part satisfies $|d|\le E(e^x+1)$, while the angular mean of $f$ vanishes. Hence the integral of the first term on the right side of Eq.~\eqref{eq:local-entropy-identity} is at least
\begin{equation}
 -E\left[\frac{U^2}{2}+(\ln\pi-1)U+\ln\pi\right].
\end{equation}
Using $(d+2\operatorname{Re}f)^2\le2d^2+8(\operatorname{Re}f)^2$ and angular Parseval gives
\begin{equation}
 \int_{x\le U}W_0\Phi(g)\,\dd q\,\dd p
 \le2E^2(e^U+2U)+4J_\rho,
 \quad
 J_\rho=\sum_{j\ge1}|v_j|^2\frac{2^j}{j!}\int_0^{L-1}e^{-x}x^j\,\dd x.
 \label{eq:Jrho}
\end{equation}
The exterior entropy of $W_\rho$ is nonnegative because $W_\rho\le1/\pi<1$. The omitted vacuum entropy is $e^{-U}(U+1+\ln\pi)$. Combining these estimates gives, for all sufficiently small $E$,
\begin{equation}
 d_\rho:=1+\ln\pi-h(W_\rho)
 \le4E(L+2)^2+4J_\rho.
 \label{eq:statewise-entropy-bound}
\end{equation}
The constants and the small-$E$ threshold in Eq.~\eqref{eq:statewise-entropy-bound} are independent of $\rho$.

\subsection{A uniform energy-weighted displacement estimate}

\begin{lemma}
\label{lem:matrix-elements}
There are absolute constants $C,R_0>0$ such that for $R\ge R_0$, all integers $k\ge1$, and $0\le j\le2R/5$,
\begin{equation}
 \frac{|\langle k|D(\sqrt{2R})|k+j\rangle|}{\sqrt{k(k+j)}}
 \le C R^{-4/3}.
 \label{eq:matrix-uniform}
\end{equation}
\end{lemma}

\begin{proof}
Write $T=2R$ and $\ell=k+j$. The exact matrix element is
\begin{equation}
 M_{k,\ell}(T):=|\langle k|D(\sqrt T)|\ell\rangle|
 =e^{-T/2}\sqrt{\frac{k!}{\ell!}}T^{j/2}|L_k^{(j)}(T)|.
 \label{eq:associated-Laguerre}
\end{equation}
First suppose $k\le T/20$. Then $\ell\le T/4$. Comparing each term of the finite Laguerre expansion with its leading term gives
\begin{equation}
 |L_k^{(j)}(T)|\le\frac{T^k}{k!}\exp(k\ell/T).
\end{equation}
Indeed, the absolute ratio for a descent of $h$ degrees is $(k)_h(\ell)_h/(h!T^h)\le(k\ell/T)^h/h!$. From $m!\ge(m/e)^m$,
\begin{equation}
 \ln M_{k,\ell}(T)\le-\frac T2+\frac k2\ln\frac{eT}{k}
 +\frac\ell2\ln\frac{eT}{\ell}+\frac{k\ell}{T}
 <-\frac{3T}{40}.
 \label{eq:small-k}
\end{equation}
For the last comparison use that $u\ln(e/u)$ increases on $(0,1)$, together with $\ln20<3$ and $\ln4<7/5$. This exponential estimate proves Eq.~\eqref{eq:matrix-uniform} in the small-$k$ range.

For $k\ge T/20$, use the Bargmann representation. Displacement maps $z^\ell/\sqrt{\ell!}$ to
\begin{equation}
 e^{-T/2+\sqrt Tz}\frac{(z-\sqrt T)^\ell}{\sqrt{\ell!}}.
\end{equation}
Extract its $z^k$ coefficient on $|z|=\sqrt k$ and multiply by $\sqrt{k!}$. The resulting integral has the form
\begin{equation}
 \langle k|D(\sqrt T)|\ell\rangle
 =\frac1{2\pi}\int_{-\pi}^{\pi}a(\theta)e^{iT\varphi(\theta)}\,\dd\theta.
 \label{eq:bargmann-circle}
\end{equation}
Put $u=\cos\theta$, $d(u)=k+T-2\sqrt{kT}u$, and $u_0=(T-j)/(2\sqrt{kT})$. The logarithm of the amplitude is
\begin{equation}
 F(u)=-\frac T2+\sqrt{kT}u+\frac\ell2\ln d(u)
       -\frac k2\ln k+\frac12\ln\frac{k!}{\ell!}.
 \label{eq:amplitude-log}
\end{equation}
Its extended stationary point satisfies $d(u_0)=\ell$. The two-sided factorial bounds
$\sqrt m(m/e)^m\le m!\le5\sqrt m(m/e)^m$ show that
\begin{equation}
 e^{F(u_0)}\le\sqrt5(k/\ell)^{1/4}\le\sqrt5.
\end{equation}
For points between $u\in[-1,1]$ and $u_0$, the affine function $d$ is positive and at most $(\sqrt k+\sqrt T)^2$, since $\ell\le k+T$. Thus
\begin{equation}
 F''(u)=-\frac{2\ell kT}{d(u)^2}
 \le-\frac{2T}{(1+\sqrt{20})^4}.
\end{equation}
With $a_0=(1+\sqrt{20})^{-4}$, Taylor's formula gives
\begin{equation}
 a(\theta)\le\sqrt5\exp[-a_0T(\cos\theta-u_0)^2].
 \label{eq:amplitude-envelope}
\end{equation}
The bound extends continuously through a zero of $d$. Strict concavity of $F$ implies that $a$ has at most one maximum on each half-circle, so its total variation on every subinterval of a half-circle is at most $2\sqrt5$.

For a fixed $\delta_0>0$, if $u_0\le1-\delta_0$, then on $|\cos\theta-u_0|\le\delta_0/2$ the absolute value of $\sin\theta$ has a uniform positive lower bound. Changing variables to $u$ in Eq.~\eqref{eq:amplitude-envelope} gives $O(T^{-1/2})$; the complementary part is exponentially small. If $u_0\ge1+\delta_0$, the full integral is exponentially small.

It remains to choose $\delta_0$ and estimate $|u_0-1|\le\delta_0$. Let $p=k/T$ and $\omega=j/T\in[0,1/5]$. In this region, initially choose $|u_0-1|\le1/10$. Then $4/11\le\sqrt p\le5/9$, so the continuous branch
\[
 \arg(\sqrt p e^{i\theta}-1)=\pi-\arctan\!\left(\frac{\sqrt p\sin\theta}{1-\sqrt p\cos\theta}\right)
\]
has denominator at least $4/9$. It gives the phase
\begin{equation}
 \varphi(\theta;p,\omega)=\sqrt p\sin\theta
 +(p+\omega)\arg(\sqrt p e^{i\theta}-1)-p\theta.
\end{equation}
Its first derivative is $H(\cos\theta)$, where
\begin{equation}
 H(u)=\sqrt p\,u+(p+\omega)\frac{p-\sqrt p\,u}{p+1-2\sqrt p\,u}-p
 =\frac{-2p(u-u_0)(u-\sqrt p)}{p+1-2\sqrt p\,u}.
\end{equation}
On $u_0=1$, $\sqrt p=(1-\omega)/2$, and therefore
\begin{equation}
 \varphi'''(0;p,\omega)=\frac{(1-\omega)^2}{1+\omega}\ge\frac8{15}.
 \label{eq:cubic-phase-derivative}
\end{equation}
Compactness of this parameter box gives constants $\theta_0,\delta_0>0$, independent of $k,j,T$, such that $\varphi'''\ge1/4$ for $|\theta|\le\theta_0$ and $|u_0-1|\le\delta_0$. Reduce $\delta_0$ further so that $\delta_0<(1-\cos\theta_0)/2$.

To obtain the oscillatory bound, first consider an interval on which $|\varphi''|\ge\mu>0$. The derivative $\varphi'$ is monotone, and the set $|\varphi'|\le\sqrt{\mu/T}$ has length at most $2(T\mu)^{-1/2}$. On each of the at most two remaining intervals, integration by parts bounds the integral of $e^{iT\varphi}$ by $C(T\mu)^{-1/2}$: the boundary terms are bounded by $2/(T\sqrt{\mu/T})$, and the remaining term is bounded by $T^{-1}\operatorname{Var}(1/\varphi')\le C/(T\sqrt{\mu/T})$. This proves the second-derivative estimate, also on any subinterval.

On $|\theta|\le\theta_0$, remove the set $|\varphi''|\le T^{-1/3}$. Its length is $O(T^{-1/3})$ because $\varphi'''\ge1/4$. The complement has at most two intervals; the preceding estimate with $\mu=T^{-1/3}$ gives $O(T^{-1/3})$ on each. The same argument holds on every subinterval of $[-\theta_0,\theta_0]$. Stieltjes integration by parts therefore gives the bound $CT^{-1/3}(\sup a+\operatorname{Var}(a))$ for the weighted integral. The amplitude factor is bounded uniformly by Eq.~\eqref{eq:amplitude-envelope} and the variation estimate above. On $|\theta|>\theta_0$, Eq.~\eqref{eq:amplitude-envelope} is exponentially small by the choice of $\delta_0$. Consequently $M_{k,\ell}(T)\le CT^{-1/3}$ for all $k\ge T/20$. Dividing by $\sqrt{k\ell}\ge k\ge T/20$ and combining with Eq.~\eqref{eq:small-k} proves the lemma.
\end{proof}

\subsection{Fourier coefficients and the universal upper bound}

For the excited block $A$ in Eq.~\eqref{eq:energy-block}, define its circular Fourier coefficients by
\[
 a_j(R)=\frac{1}{2\pi}\int_0^{2\pi}
 \pi W_A(\sqrt R\cos\theta,\sqrt R\sin\theta)e^{-ij\theta}\,\dd\theta.
\]
Lemma~\ref{lem:matrix-elements} and $|A_{k,k+j}|\le\sqrt{A_{kk}A_{k+j,k+j}}$ give
\begin{align}
 |a_j(R)|&\le CR^{-4/3}\sum_{k\ge1}\sqrt{k(k+j)}|A_{k,k+j}|\notag\\
 &\le CE R^{-4/3},\qquad 0\le j\le2R/5.
 \label{eq:Fourier-energy}
\end{align}
For $j=0$, this uses $\Tr(\hat N A)\le E$ directly; for $j>0$, apply Cauchy--Schwarz to the energy-weighted diagonal entries. The sums are absolutely convergent. Trace-norm convergence of finite Fock compressions and the parity bound justify the Fourier identities for infinite-support $A$.

The nonnegative circular function $\pi W_\rho$ has each nonzero Fourier coefficient bounded in modulus by its constant coefficient. Equation~\eqref{eq:wigner-block} therefore implies
\begin{equation}
 |v_j|\sqrt{\frac{2^j}{j!}}R^{j/2}
 \le1+e^R\bigl(|a_0(R)|+|a_j(R)|\bigr)
 \le1+2CEe^R R^{-4/3}.
\end{equation}
Take
\begin{equation}
 R=L+\frac43\ln L.
 \label{eq:enlarged-radius}
\end{equation}
Since $Ee^RR^{-4/3}=(L/R)^{4/3}\le1$, there is an absolute constant $M$ for which
\begin{equation}
 |v_j|^2\frac{2^j}{j!}\le M^2R^{-j},\qquad1\le j\le2R/5.
 \label{eq:new-coefficient}
\end{equation}
Higher harmonics retain the separate bound in Eq.~\eqref{eq:old-coefficient}.

Let $s=\lfloor r_*L+c_*\ln L\rfloor$ and $m=\lfloor2L/5\rfloor$. For sufficiently large $L$, $1<s<m<2R/5$. Split $J_\rho$ in Eq.~\eqref{eq:Jrho} into $j<s$, $s\le j\le m$, and $j>m$. The first part satisfies
\begin{equation}
 J_{\mathrm{low}}\le E\max\left(2,\frac{2^{s-1}}{s-1}\right)
 \le2E+CE^\gamma L^{c_*\ln2-1},
 \label{eq:Jlow}
\end{equation}
because $j\mapsto j\ln2-\ln j$ is convex and Eq.~\eqref{eq:coherence-energy} controls the weighted sum.

Using Eq.~\eqref{eq:new-coefficient} and enlarging a sum of positive terms to a geometric series,
\begin{equation}
 J_{\mathrm{mid}}\le M^2R^{-s}\int_0^{L-1}\frac{e^{-x}x^s}{1-x/R}\,\dd x
 \le C\frac{s!}{R^s}.
 \label{eq:Jmid}
\end{equation}
To prove the last estimate uniformly, use $(1-x/R)^{-1}\le4$ on $x\le3L/4$ and at most $2L$ on the remaining interval. For $X$ with the gamma distribution of shape $s+1$ and scale one,
\begin{equation}
 \Pr(X\ge3L/4)\le2^{s+1}e^{-3L/8}
 \le2e^{-19L/200},
\end{equation}
where $s\le2L/5$ eventually and $\ln2<7/10$. Hence the integral in Eq.~\eqref{eq:Jmid} is at most $5s!$ for sufficiently large $L$.

Stirling's formula, including its square-root factor, now gives
\begin{align}
 \ln\frac{s!}{R^s}
 &=-\gamma L+\left(\frac12+c_*\ln r_*-\frac43r_*\right)\ln L+O(1)\notag\\
 &=-\gamma L-\beta\ln L+O(1).
 \label{eq:stirling-upper}
\end{align}
Here $r_*(1-\ln r_*)=\gamma$ and $c_*b_*+4r_*/3=3/2$. The second-order expansion error is $O((\ln L)^2/L)$; the integer part contributes only a bounded term.

For the last frequency range, Eq.~\eqref{eq:old-coefficient} gives
\begin{equation}
 J_{\mathrm{high}}\le4L\frac{(m+1)!}{L^{m+1}}
 =O\!\left(L^{3/2}E^{g_*}\right),\qquad
 g_*:=\frac25\left[1-\ln(2/5)\right]>\gamma.
 \label{eq:Jhigh}
\end{equation}
The strict inequality follows from the increase of $u(1-\ln u)$ on $[r_*,2/5]$. Since $\gamma<1$, both this term and $E(L+2)^2$ are $o(E^\gamma L^{-\beta})$. Equations~\eqref{eq:statewise-entropy-bound} and \eqref{eq:Jlow}--\eqref{eq:Jhigh} give one bound, uniform over every admissible $\rho$. Taking the supremum proves the upper half of Eq.~\eqref{eq:sharp-energy}.

\subsection{A positive repair with controlled energy}

For $R\ge1$, set $v=R^{-1/3}$. Let $\chi_R$ be the uniform phase average of a pure Gaussian state centred at $(\sqrt R,0)$ with covariance $\operatorname{diag}(v/2,1/(2v))$. Mixtures of displaced states squeezed along a circle's radius were considered in Ref.~\cite[Appendix 2, Eq.~(19)]{MKC2008v1}. In the present convention,
\begin{equation}
 W_{\chi_R}(x)=\frac1{2\pi^2}\int_0^{2\pi}
 \exp\!\left[-\frac{(\sqrt x\cos\theta-\sqrt R)^2}{v}-vx\sin^2\theta\right]\dd\theta,
 \label{eq:ring-Wigner}
\end{equation}
\begin{equation}
 N_{\chi_R}=\frac{R+(v+v^{-1})/2-1}{2}\le R.
 \label{eq:ring-energy}
\end{equation}
The determinant of the covariance is $1/4$, so each component is a normalised pure Gaussian state; the phase average is a physical, radial, strictly Wigner-positive state.

For $x=R+y$, $y\ge0$, choose $\theta_*$ with $\sqrt x\cos\theta_*=\sqrt R$ and $\sqrt x\sin\theta_*=\sqrt y$. Put $w=R^{1/6}$, $d=\sqrt y$, and $h=[2w(d+w)]^{-1}$. For $|\theta-\theta_*|\le h$, the elementary trigonometric estimates give
\begin{equation}
 \frac{|\sqrt x\cos\theta-\sqrt R|}{\sqrt v}
 \le\frac{d}{2(d+w)}+\frac{w^2}{8(d+w)^2}\le\frac58,
\end{equation}
\begin{equation}
 vx\sin^2\theta\le vy+\frac{d}{d+w}+\frac{w^2}{4(d+w)^2}\le vy+\frac54.
\end{equation}
Also $h\ge R^{-1/3}/[2\sqrt2\sqrt{1+y/R^{1/3}}]$. Integrating over this periodic window of length $2h$ proves the full exterior bound
\begin{equation}
 W_{\chi_R}(R+y)\ge\frac{e^{-2}}{2\sqrt2\pi^2}R^{-1/3}
 \left(1+\frac{y}{R^{1/3}}\right)^{-1/2}e^{-R^{-1/3}y},\qquad y\ge0.
 \label{eq:ring-lower}
\end{equation}

Fix a sufficiently large constant $K$, and for each small $E$ choose
\begin{equation}
 n=\lfloor r_*L+c_*\ln L-K\rfloor,\qquad a=\frac{E}{4n},\qquad
 R=L+\frac43\ln L,\qquad
 A_R=2\sqrt a\sqrt{\frac{2^n}{n!}}R^{n/2}.
 \label{eq:construction-parameters}
\end{equation}
If $\vartheta_L=r_*L+c_*\ln L-K-n\in[0,1)$, Stirling's formula gives
\begin{equation}
 \ln A_R^2=-b_*K-b_*\vartheta_L-\frac32\ln r_*-\frac12\ln(2\pi)+o_K(1).
 \label{eq:construction-amplitude}
\end{equation}
Thus $K$ can be fixed first so that $A_R\le1/4096$ for all sufficiently large $L$. We also have $n\ge3$, $2n<R$, and $n\le2R/5$ eventually. Define
\begin{equation}
 \epsilon=256A_RR^{1/3}e^{-R},\qquad
 |\psi_E\rangle=\frac{|0\rangle+\sqrt a|n\rangle}{\sqrt{1+a}},\qquad
 \rho_E=(1-\epsilon)|\psi_E\rangle\langle\psi_E|+\epsilon\chi_R.
 \label{eq:attaining-state}
\end{equation}
Here $\epsilon\to0$ and $\limsup\epsilon/a\le1024r_*/4096<1$. By Eq.~\eqref{eq:ring-energy}, all displacement and squeezing costs obey
\begin{equation}
 \Tr(\hat N\rho_E)\le na+\epsilon R
 \le\frac E4+256A_RE\left(\frac RL\right)^{4/3}\le\frac E2
 \label{eq:total-energy}
\end{equation}
for sufficiently small $E$.

We next establish strict Wigner positivity. The Fock polynomial satisfies $|F_n(x)|\le e^x$ by the parity bound, and
\begin{equation}
 F_n(x)>0\qquad(x\ge2n).
 \label{eq:Laguerre-tail-positive}
\end{equation}
To see the latter without a zero asymptotic, $(-1)^nn!L_n(y)$ is the characteristic polynomial of the symmetric tridiagonal matrix with diagonal $1,3,\ldots,2n-1$ and off-diagonal entries $1,2,\ldots,n-1$. Its eigenvalues are below $4n$, by the row-sum bound, so its characteristic polynomial is positive for $y\ge4n$.

For $x\le2n$, the numerator of $W_{|\psi_E\rangle\langle\psi_E|}/W_0$ is at least $1-ae^{2n}-A_R>0$, since $ae^{2n}\to0$ and $2r_*<1$. For $2n<x\le R$, Eq.~\eqref{eq:Laguerre-tail-positive} and $A_R<1$ again give positivity. For $x=R+y\ge R$, the negative part can only come from the off-diagonal term. Relative to Eq.~\eqref{eq:ring-lower}, its modulus is at most
\begin{equation}
 \frac{2\sqrt2\pi e^2}{1+a}A_R R^{1/3}e^{-R}
 \left(1+\frac yR\right)^{n/2}
 \left(1+\frac{y}{R^{1/3}}\right)^{1/2}e^{-(1-R^{-1/3})y}.
 \label{eq:ring-ratio}
\end{equation}
Using $\ln(1+u)\le u$, the logarithm of its $y$-dependent factor is at most
\begin{equation}
 \left[\frac{n}{2R}+\frac32R^{-1/3}-1\right]y\le-\frac{27}{40}y,
\end{equation}
for $R\ge2048$. Since $2\sqrt2\pi e^2<108<128$, Eq.~\eqref{eq:attaining-state} therefore gives, both inside and outside the circle,
\begin{equation}
 W_{\rho_E}>\frac\epsilon2 W_{\chi_R}>0.
 \label{eq:attaining-positive}
\end{equation}
For each fixed $E$, the polynomial-Gaussian part and the compact phase average in Eq.~\eqref{eq:ring-Wigner} have Gaussian tail bounds. Together with Eq.~\eqref{eq:attaining-positive}, these ensure finite entropy and finite relative entropy to $W_0$.

\subsection{Entropy of the attaining family}

Put $c_0=(1-\epsilon)/(1+a)$, which is at least $15/16$ for all sufficiently small $E$, and use the bounded test function
\begin{equation}
 u(q,p)=2\sqrt a\sqrt{\frac{2^n}{n!}}x^{n/2}\cos(n\theta)\,\boldsymbol{1}_{\{x\le2n\}},
 \qquad |u|\le A_R\le\frac14.
 \label{eq:bounded-test}
\end{equation}
All terms in $W_{\rho_E}$ except the vacuum--$n$ coherence are radial. Angular orthogonality gives
\begin{equation}
 \int uW_{\rho_E}=c_0V_0,\qquad V_0=\int u^2W_0.
\end{equation}
A shift of $\theta$ by $\pi/n$ changes $u$ to $-u$, hence $\int e^uW_0=\int\cosh(u)W_0$. For $|u|\le1/4$, $\cosh u-1\le(8/15)u^2$, and thus
\begin{equation}
 \ln\int e^uW_0\le\frac8{15}V_0.
\end{equation}
The variational relative-entropy inequality follows by applying nonnegativity of relative entropy to the density $e^uW_0/\int e^uW_0$. It yields
\begin{equation}
 D_{\mathrm{KL}}(W_{\rho_E}\|W_0)
 \ge\left(c_0-\frac8{15}\right)V_0\ge\frac38V_0.
 \label{eq:Gibbs-lower}
\end{equation}
Direct angular integration gives
\begin{equation}
 V_0=2a2^n\frac1{n!}\int_0^{2n}e^{-x}x^n\,\dd x\ge\frac23a2^n.
\end{equation}
The last inequality follows from Markov's inequality for a gamma variable with mean $n+1$: its probability of being at most $2n$ is at least $(n-1)/(2n)\ge1/3$. Consequently $D_{\mathrm{KL}}\ge a2^n/4$. By Eq.~\eqref{eq:total-energy},
\begin{equation}
 h(W_{\rho_E})-(1+\ln\pi)
 =2\Tr(\hat N\rho_E)-D_{\mathrm{KL}}(W_{\rho_E}\|W_0)
 \le E-\frac{a2^n}{4}\le-\frac{a2^n}{8},
\end{equation}
where the last comparison holds when $2^n\ge32n$. This condition is eventually satisfied because $n/L\to r_*>0$. Finally,
\begin{equation}
 \frac{a2^n}{8}=\frac{E2^n}{32n}
 \ge C_1E^\gamma L^{c_*\ln2-1}
 =C_1E^\gamma L^{-\beta}
\end{equation}
for a positive constant $C_1$. All requirements hold for every sufficiently small $E$, with $K$ fixed first. This proves the lower half of Eq.~\eqref{eq:sharp-energy} and completes Theorem~\ref{thm:sharp-energy}.

\section{Derivation of the sharp Shannon loss threshold}
\label{app:loss}

We include the short proof from the companion stability theory \cite{HeStability} in the notation of the present counterexamples.
At half transmission,
\[
W_{\mathcal L_{1/2}^{\otimes m}(\rho)}(q,p)
=\pi^{-m}\langle q+ip|\rho|q+ip\rangle,
\]
where the ket is a multimode coherent state with amplitude vector $q+ip$.
The Wehrl inequality gives entropy at least $m(1+\ln\pi)$ \cite{Wehrl}. For $\eta<1/2$, apply the same result to $\mathcal L_{2\eta}^{\otimes m}(\rho)$. The vacuum saturates the bound. This argument applies to every density operator; under finite input energy the entropies are finite.

Now fix $\eta>1/2$ and $0<\lambda<1$. For each fixed Fock level $n$, take $\sigma_{n,t,\lambda}$ from Eq.~\eqref{eq:doublet-state} and put $a=t^2$. Theorem~\ref{thm:general-doublet} ensures strict input Wigner positivity for all sufficiently small $a$. Pure loss preserves Wigner positivity. The output diagonal Wigner polynomial is
\[
P_{n,\eta}(x)=\sum_{j=0}^n\binom nj(1-2\eta)^{n-j}
\frac{(2\eta x)^j}{j!},
\]
and the vacuum--$n$ coherence is multiplied by $\eta^{n/2}$.
Write $W_{\rm out}=W_0(1+u)$, with $\int W_0u=0$.
The global scalar bound
\[
\big|(1+u)\ln(1+u)-u-u^2/2\big|\le C|u|^3,\qquad u\ge-1,
\]
follows by treating $|u|\le1/2$ and its complement separately, using the continuous value at $u=-1$.
At fixed $n,\eta,\lambda$, polynomial Gaussian moments then give
\[
\mathcal D_1(\mathcal L_\eta\sigma_{n,\sqrt a,\lambda})
=a\{\lambda^2(2\eta)^n-2\eta n\}
+O_{n,\eta,\lambda}(a^{3/2}).
\]
The input energy per mode is $e=na/(1+a)$. Choose
$M=\lfloor E/e\rfloor$ independent factors. Their complete energy is at most $E$, and, with $n$ fixed first, their total output deficit tends as $a\downarrow0$ to
\[
E\left\{\frac{\lambda^2(2\eta)^n}{n}-2\eta\right\}.
\]
This limit is unbounded as $n$ grows because $2\eta>1$.
For any desired finite deficit, choose $n$ first and then $a$ small enough to satisfy positivity, energy and remainder bounds simultaneously. Every chosen state has finitely many modes and finite Fock support. This proves the second half of Eq.~\eqref{eq:main-loss}.

\section{An explicit source with residual structure after entropy recovery}
\label{app:residual}

We give the source and witness argument underlying Eq.~\eqref{eq:main-hull-gap}, specialised from Refs.~\cite{HeConcentration,HeStability}.
Fix $0<c<1$, put $p=c/m$, and define
\begin{align}
|\phi_p\rangle&=\sqrt{1-p}|0\rangle+\sqrt p|3\rangle,\notag\\
R_p&=600\exp[-(64p)^{-1/3}/4],\qquad
\epsilon_p=\frac{R_p}{1+R_p},\notag\\
\omega_p&=(1-\epsilon_p)|\phi_p\rangle\langle\phi_p|+\epsilon_p\tau,
\qquad
\widetilde\rho_{m,\eta}=\mathcal L_\eta^{\otimes m}(\omega_p^{\otimes m}),
\label{eq:residual-source}
\end{align}
with the same mean-$1/2$ thermal state $\tau$ as in Eq.~\eqref{eq:main-thermal}.
The complete input energy is
$3c(1-\epsilon_p)+m\epsilon_p/2\to3c$, and
$T(\omega_p^{\otimes m},|\phi_p\rangle\langle\phi_p|^{\otimes m})\le m\epsilon_p\to0$ faster than every inverse power of $m$.

For $p\le1/16$, the angular minimum of the pure Wigner polynomial is
\[
F_{p,\min}(u)=1-2p+6pu-6pu^2+\frac43pu^3
-\frac4{\sqrt3}\sqrt{p(1-p)}u^{3/2}.
\]
On $0\le u\le u_0=(64p)^{-1/3}$, discarding positive terms gives the lower bound
$1-1/8-(3/8)16^{-1/3}-1/(2\sqrt3)>0$.
Everywhere, the absolute polynomial is at most $3(1+u^3)$.
On the remaining region, its negative Wigner part divided by $W_\tau$ is therefore at most
\[
6(1+u^3)e^{-u/2}
\le6[1+(12/e)^3]e^{-u_0/4}<R_p.
\]
The odds $\epsilon_p/(1-\epsilon_p)=R_p$ establish global positivity of $\omega_p$. Products and pure loss preserve it.

Let $|\Phi_m\rangle=m^{-1/2}\sum_{j=1}^m|3_j\rangle$, where $|3_j\rangle$ has three photons in mode $j$ and vacuum elsewhere, and set $P_m=|\Phi_m\rangle\langle\Phi_m|$.
We first prove the Gaussian bound
\begin{equation}
\sup_{\sigma\in\mathcal F_m}\Tr(P_m\sigma)\le\frac{18}{m}.
\label{eq:residual-Gaussian}
\end{equation}
A normalised pure Gaussian state has Bargmann function
$f(z)=\sqrt C\exp(\tfrac12z^TAz+d^Tz)$, where $A=A^T$ and its Takagi singular values $s_j$ lie below one.
Writing $b$ for $d$ in the Takagi basis, Gaussian integration gives
\[
C^{-1}=\prod_j(1-s_j^2)^{-1/2}
\exp\sum_j\left[\frac{(\Re b_j)^2}{1-s_j}
+\frac{(\Im b_j)^2}{1+s_j}\right]
\ge\exp[(\|A\|_{\mathrm{HS}}^2+\|d\|^2)/2].
\]
The three-photon coefficient is
$\langle3_j|f\rangle=\sqrt{C/6}(d_j^3+3A_{jj}d_j)$.
With $x=\|A\|_{\mathrm{HS}}$ and $y=\|d\|$,
\[
|\langle\Phi_m|f\rangle|^2
\le\frac{C}{6m}(y^3+3xy)^2
\le\frac1m\left(\frac{72}{e^3}+\frac{12}{e^2}\right)
<\frac{18}{m}.
\]
Here $\sum_j|d_j|^3\le y^3$, $\sum_j|A_{jj}d_j|\le xy$,
and the two scalar maxima follow from
$C\le e^{-(x^2+y^2)/2}$.
Every mixed Gaussian state is a Gaussian displacement mixture of pure Gaussian states. Linearity and trace-norm continuity of the bounded effect extend the estimate to $\mathcal F_m$, including its closure and unrestricted comparator energy.

For the unrepaired product seed after loss, put
\[
q=1-p[1-(1-\eta)^3],\qquad
b=p\eta^3,\qquad z=\eta^{3/2}\sqrt{p(1-p)}.
\]
Its witness probability is
\[
\Tr(P_m\rho_{\rm seed,out})
=bq^{m-1}+(m-1)|z|^2q^{m-2}
\longrightarrow c\eta^3e^{-c[1-(1-\eta)^3]}.
\]
The physical repair changes this probability by at most $m\epsilon_p$, because loss contracts trace distance. For every $\sigma\in\mathcal F_m$,
\[
T(\widetilde\rho_{m,\eta},\sigma)
\ge\Tr[P_m(\widetilde\rho_{m,\eta}-\sigma)]
\ge bq^{m-1}+(m-1)|z|^2q^{m-2}-m\epsilon_p-\frac{18}{m}.
\]
Taking the infimum and then the lower limit proves Eq.~\eqref{eq:main-hull-gap}.
At $\eta=1/2$, Appendix~\ref{app:loss} applies to this very same physical output, giving entropy recovery alongside the positive limiting structural gap.

\subsection{The concentration theorem and its transfer through loss}
\label{app:retention-interface}

We state the exact Gaussian-processing interface used in Eq.~\eqref{eq:main-retention}. The source is the physical state in Eq.~\eqref{eq:residual-source}, with fixed $0<c<1$. Each attempt receives one complete $m$-mode copy. An allowed protocol uses finitely many modes and operations in each instance: normal Gaussian ancillary states, Gaussian unitaries and channels, Gaussian measurements including ideal homodyne probability instruments, arbitrary measurable classical feedback, randomisation, discarding and success selection. Measurement records lie in standard Borel spaces. The supremum ranges over all finite choices of ancillary energy, squeezing, depth and number of operations.

For its trace-nonincreasing success map $\mathcal I_{\rm s}$, define
\begin{equation}
p_{\rm s}=\Tr\mathcal I_{\rm s}(\widetilde\rho_{m,\eta}),
\qquad
\rho_{\rm out\mid succ}=\frac{\mathcal I_{\rm s}(\widetilde\rho_{m,\eta})}{p_{\rm s}}.
\label{eq:retention-success}
\end{equation}
The successful output has exactly $k$ modes. The threshold $0<p_0\le1$ is fixed independently of $m$ and applies to $p_{\rm s}$ on this actual repaired input. With $\mathcal F_k$ the full trace-norm-closed Gaussian convex hull, set
\begin{equation}
B_{m,\eta}(k,p_0)=\sup_{\mathcal I:\,p_{\rm s}\ge p_0}
d_{\mathcal F_k}(\rho_{\rm out\mid succ}).
\label{eq:retention-B}
\end{equation}

The uniform concentration theorem of Ref.~\cite{HeConcentration}, also proved in Ref.~\cite{HeStability}, supplies constants $C_{c,p_0}$ and $m_0(c,p_0)$ such that
\begin{equation}
B_{m,1}(k,p_0)\le C_{c,p_0}
\left(\sqrt{k/m}+m^{-1/4}\right),
\qquad m\ge m_0,\quad 1\le k\le m.
\label{eq:retention-import}
\end{equation}
This is the imported processing theorem; Appendix~\ref{app:residual} supplies the source and structural witness in the present notation.

For an attenuated input, precompose the same success map with pure loss:
$\mathcal J_{\rm s}=\mathcal I_{\rm s}\circ\mathcal L_\eta^{\otimes m}$.
The composed protocol is Gaussian and acts on $\omega_{c/m}^{\otimes m}$. Its success probability and normalised successful output are exactly those in Eq.~\eqref{eq:retention-success}. Therefore
\begin{equation}
B_{m,\eta}(k,p_0)\le B_{m,1}(k,p_0)
\le C_{c,p_0}\left(\sqrt{k/m}+m^{-1/4}\right).
\label{eq:retention-loss-transfer}
\end{equation}
The same constants work for every $0\le\eta\le1$.

Fix $0<\eta\le1$ and a target
$0<g<c\eta^3e^{-c[1-(1-\eta)^3]}$.
For sufficiently large $m$, $C_{c,p_0}m^{-1/4}\le g/2$, so an output with gap at least $g$ requires
\begin{equation}
k\ge m\left(\frac{g}{2C_{c,p_0}}\right)^2.
\label{eq:retention-mode-lower}
\end{equation}
The identity protocol retains all $m$ modes with probability one and eventually has gap greater than $g$ by Eq.~\eqref{eq:main-hull-gap}. These two bounds establish the $\Theta(m)$ minimum mode count for this source, target gap and fixed actual success threshold. In particular they apply to the same half-loss output whose Shannon entropy has recovered the vacuum bound.

\end{document}